%% file: main.tex
\documentclass[sigconf]{acmart}
\input{macros.tex}

\begin{document}

\copyrightyear{2017}
\acmYear{2017}
\setcopyright{acmlicensed}
\acmConference{CIKM'17 }{November 6--10, 2017}{Singapore,
Singapore}
\acmPrice{15.00}
\acmDOI{10.1145/3132847.3132991}
\acmISBN{978-1-4503-4918-5/17/11}

\fancyhead{}
\settopmatter{printacmref=false, printfolios=false}

\title[The Minimum Inefficiency Subgraph Problem]{To Be Connected, or Not to Be  Connected:\\  That is the Minimum Inefficiency Subgraph Problem}

\author{Natali Ruchansky}
\affiliation{\small Univ. of Southern California, USA}
\email{natalir@bu.edu}
\author{Francesco Bonchi}
\affiliation{\small ISI Foundation, Italy}
\email{francesco.bonchi@isi.it}
\author{David Garc\'ia-Soriano}
\affiliation{\small Universidad Pompeu Fabra, Spain}
\email{elhipercubo@gmail.com }
\author{Francesco Gullo}
\affiliation{\small UniCredit, R\&D Dept., Italy}
\email{gullof@acm.org }
\author{Nicolas Kourtellis}
\affiliation{\small Telefonica Research, Spain}
\email{nicolas.kourtellis@telefonica.com }

\renewcommand{\shortauthors}{N. Ruchansky et al.}

\begin{abstract}
We study the problem of extracting a \emph{selective connector} for a given set of query vertices $Q \subseteq V$ in a graph $G = (V,E)$. A selective connector is a subgraph of $G$ which exhibits some cohesiveness property,  and contains the query vertices but does not necessarily connect them all. Relaxing the connectedness requirement allows the connector to detect multiple communities and to be tolerant to outliers. We achieve this by introducing the new measure of \emph{network inefficiency} and by instantiating our search for a selective connector as the problem of finding the \emph{minimum inefficiency subgraph}.

We show that the \emph{minimum inefficiency subgraph} problem is \NPhard, and devise efficient algorithms to approximate it. By means of several case studies in a variety of application domains (such as human brain, cancer, and food networks), we show that our minimum inefficiency subgraph produces high-quality solutions, exhibiting all the desired behaviors of a selective connector.
\end{abstract}


\maketitle
\sloppy

\section{Introduction}
\label{sec:intro}
\input{sections/intro.tex}

\section{Related work}
\label{sec:related}
\input{sections/related.tex}

\section{Minimum inefficiency subgraph}
\label{sec:problem}
\input{sections/problem.tex}

\section{Algorithms}
\label{sec:algorithms}
\input{sections/algorithms.tex}

\section{Experiments}
\label{sec:experiments}
\input{sections/experiments.tex}

\section{Case studies}
\label{sec:casestudies}

\input{sections/casestudies}

\section{Conclusions}
\label{sec:conclusions}
\input{sections/conclusions.tex}

\bibliographystyle{abbrv}
\bibliography{references}

\end{document}

%% file: macros.tex
\usepackage{graphicx}
\usepackage{amsmath,amssymb}
\usepackage{epsfig}
\usepackage{subcaption}
\usepackage{url}
\usepackage{hyperref}
\usepackage{adjustbox}
\usepackage{booktabs}     
\usepackage{multirow}
\usepackage{siunitx}	
\usepackage{rotating}
\usepackage{enumerate}
\usepackage{algorithm}
\usepackage{algpseudocode}
\usepackage{subcaption}
\usepackage[multiple]{footmisc}
\usepackage{float}

\algnewcommand{\LineComment}[1]{\Statex \(\triangleright\) #1}
\algnewcommand{\LineCommentSpaced}[2]{\Statex \(#2 \triangleright\) #1}

\newcolumntype{R}[2]{%
    >{\adjustbox{angle=#1,lap=\width-(#2)}\bgroup}%
    l%
    <{\egroup}%
}

\newcommand{\spara}[1]{\smallskip\noindent{\bf #1}}

\newtheorem{mydefinition}{Definition}
\newtheorem{myexample}{Example}
\newtheorem{problem}{Problem}

\DeclareMathOperator*{\argmin}{arg\,min}

\newcommand{\NPhard}{$\mathbf{NP}$-hard}

\newcommand{\squishlist}{
 \begin{list}{$\bullet$}
  {  \setlength{\itemsep}{0pt}
     \setlength{\parsep}{3pt}
     \setlength{\topsep}{3pt}
     \setlength{\partopsep}{0pt}
     \setlength{\leftmargin}{2em}
     \setlength{\labelwidth}{1.5em}
     \setlength{\labelsep}{0.5em}
} }
\newcommand{\squishlisttight}{
 \begin{list}{$\bullet$}
  { \setlength{\itemsep}{0pt}
    \setlength{\parsep}{0pt}
    \setlength{\topsep}{0pt}
    \setlength{\partopsep}{0pt}
    \setlength{\leftmargin}{2em}
    \setlength{\labelwidth}{1.5em}
    \setlength{\labelsep}{0.5em}
} }

\newcommand{\squishdesc}{
 \begin{list}{}
  {  \setlength{\itemsep}{0pt}
     \setlength{\parsep}{3pt}
     \setlength{\topsep}{3pt}
     \setlength{\partopsep}{0pt}
     \setlength{\leftmargin}{1em}
     \setlength{\labelwidth}{1.5em}
     \setlength{\labelsep}{0.5em}
} }

\newcommand{\squishend}{
  \end{list}
}

\newcommand{\cps}{\textbf{\textsf{cps}}}
\newcommand{\ctp}{\textbf{\textsf{ctp}}}
\newcommand{\mwc}{\textbf{\textsf{mwc}}}
\newcommand{\bh}{\textbf{\textsf{ldm}}}
\newcommand{\mdl}{\textbf{\textsf{mdl}}}

\newcommand{\mis}{\textbf{\textsf{mis}}}

\newcommand{\greedy}{\textbf{\textsf{GRA\_mis}}}
\newcommand{\exh}{\textbf{\textsf{exh}}}
\newcommand{\greedycps}{\textbf{\textsf{GRA\_cps}}}
\newcommand{\greedyctp}{\textbf{\textsf{GRA\_ctp}}}

\providecommand{\poly}{{\operatorname{poly}}}

\newcommand{\dataset}[1]{\textsf{#1}}

\newcommand{\mycomment}[1]{}

\newenvironment{prooftext}[1]{\par\noindent{\bf Proof#1.}\quad}{\nopagebreak$\qed$\\}
\newenvironment{myproof}{\begin{prooftext}{}}{\nopagebreak\end{prooftext}}

\newcommand{\ouralg}{\text{Algorithm~\ref{algo:greedy}}}

\newcommand{\eff}[1]{\mathcal{E}({#1})}

%% file: sections/intro.tex
\enlargethispage*{\baselineskip}
Finding subgraphs connecting a given set of  vertices of interest is a fundamental graph mining primitive which has received a great deal of attention. The extracted substructures provide insights on the relationships that exist among the given set of vertices.
For instance, given a set of proteins in a protein-protein interaction (PPI) network, one might find other proteins that participate in pathways with them.
Given a set of users in a social network that clicked an ad, to which other users (by the principle of ``homophily'') should the same ad be shown. 
Several problems of this type have been studied under different names, e.g., \emph{community search}~\cite{SozioKDD10,SozioLocalSIGMOD14,BarbieriBGG15},
\emph{seed set expansion} \cite{Andersen1,Kloumann}, \emph{connectivity subgraphs} \cite{connect,CenterpieceKDD06,ruchansky2015minimum,akoglu2013mining}, just to mention a few.  While optimizing for different objective functions, the bulk of this literature (briefly surveyed in Section \ref{sec:related}) shares a common aspect: the solution must be a \emph{connected} subgraph of the input graph containing the set of query vertices.

The \emph{requirement of connectedness} is a strongly restrictive one. Consider, for example, a biologist inspecting a set of proteins that she \emph{suspects} could be cooperating in some biomedical setting.  It may very well be the case that one of the proteins is not related to the others: in this case, forcing the sought subgraph to connect them all might produce poor quality solutions, while at the same time hiding an otherwise good solution. By relaxing the connectedness condition, the outlier protein can be kept disconnected, thus returning a much better solution to the biologist.

Another consequence of the connectedness requirement is that by trying to connect possibly unrelated vertices, the resulting solutions end up being very large. 
As highlighted in \cite{ruchansky2015minimum}, the bulk of the literature makes the (more or less implicit) assumption that the query vertices belong to the same community. When such an assumption on the query set is satisfied, these methods return reasonably compact subgraphs. However, when the query vertices belong to different modules of the input graph, these methods tend to return too large a subgraph, often so large as to be meaningless and unusable in applications.

In this paper, we study the \emph{selective connector problem}: given a graph $G = (V,E)$ and a set of query vertices $Q \subseteq V$,  find a superset $S
\supseteq Q$ of vertices such that its induced subgraph, denoted $G[S]$, has some good ``cohesiveness'' properties, but is not necessarily connected. Abstractly, we would like our selective connector $G[S]$ to have the following desirable properties:
\squishlist
\item \textbf{Parsimonious vertex addition.} Vertices should be added to $Q$ to form the solution $S$, if and only if
they help form more ``cohesive'' subgraphs by better connecting the vertices in $Q$. Roughly speaking, this ensures that
the only vertices added are those which serve to better explain the connection between the elements of $Q$ (or a subset thereof).

\item \textbf{Outlier tolerance.} If $Q$ contains vertices which are ``far'' from the rest of $Q$, those should remain
disconnected in the solution $S$ and be considered as outliers. The necessity for this stems from the fact that real-world
query-sets are likely to contain some vertices that are erroneously suspected of being related.

\item \textbf{Multi-community awareness.} If the query vertices $Q$ belong to two or more communities, then
the connector should be able to recognize this situation, detect the communities, and refrain from imposing connectedness between them.

\squishend
\enlargethispage*{\baselineskip}
So far, cohesiveness has been discussed abstractly, without a formal definition of a specific measure.
A natural way to define the cohesiveness of a subgraph $G[S]$ is to consider the shortest-path distance $d_{G[S]}(u,v)$ between every pair of vertices $u,v \in S$. Shortest paths define fundamental structural properties of networks, playing a key role in basic
mechanisms such as their evolution~\cite{kossinets2006empirical}, the formation of communities~\cite{girvan02community}, and the propagation of information; for example,
\emph{betweenness centrality}~\cite{bavelas48} which is defined as the fraction of shortest paths that a vertex participates in, is a measure of the extent to which an actor has control over information flow in the network.

One issue with shortest-path distance is that, when the connectedness requirement is dropped, pairs of vertices can be disconnected, thus yielding an infinite distance.
A simple yet elegant workaround to this issue is to use the reciprocal of the shortest-path distance \cite{Marchiori2000}; this has the useful property of
handling $\infty$ neatly (assuming by convention that $\infty^{-1} = 0$). This is the idea at the heart of \emph{network efficiency}, a graph-theoretic notion that was introduced by Latora and Marchiori \cite{latora2001efficient} as a measure of how efficiently a network $G = (V,E)$ can exchange information:

\smallskip
$$
\eff{G}=\frac{1}{|V|(|V|-1)}\sum_{\substack{u,v \in V \\ u\neq v}}\frac{1}{d_G(u,v)}.
$$
\smallskip

Unfortunately, defining the selective connector problem as finding the subgraph $G[S]$ with $S \supseteq Q$  that \emph{maximizes network efficiency} would be meaningless. In fact,
as we show in Section \ref{sec:problem}, the normalization factor $|V|(|V| - 1)$ allows vertices totally unrelated to $Q$ to be added to improve the efficiency; clearly violating our driving principle of parsimonious vertex addition.
Based on the above arguments, we introduce the measure of the \emph{inefficiency} of a graph $G = (V,E)$, defined as follows:

\smallskip
$$
\mathcal{I}(G) = \sum_{\substack{u,v \in V \\ u\neq v}}{1-\frac{1}{d_{G}(u,v)}}.
$$
\smallskip

Hence, we define the selective connector problem as the \emph{parameter-free} problem which requires extracting the subgraph $G[S]$, with $S \supseteq Q$, that \emph{minimizes network inefficiency}.
With this definition, each pair of vertices in the subgraph $G[S]$ produces a cost between 0 and 1, which is minimum when the two vertices are neighbors, grows with their
distance, and is maximum when the two vertices are not reachable from one another. Parsimony in adding vertices is handled by the sum of costs over all pairs of vertices in
the connector; adding one vertex $v$ to a partial solution $S$ incurs $|S|$ more terms in the summation.  The inclusion of $v$ is worth the additional cost only if these costs are small and if $v$ helps reduce the distances between vertices in $S$. Moreover, note that by allowing disconnections in the solution, the second and third design principles above (i.e., outliers and multiple communities) naturally follow from the parsimonious vertex addition.

The \emph{Minimum Inefficiency Subgraph} problem is \NPhard, and we prove that it remains hard even if we constrain the input graph $G$ to have a diameter of at most 3.
Therefore, we devise an algorithm that is based on first building a complete connector for the query vertices and then \emph{relaxing} the connectedness requirement by \emph{greedily} removing non-query vertices. Our experiments show that in 99\% of problem settings, our greedy relaxing algorithm produces solutions no worse than those produced by an exhaustive search, while at the same time being orders of magnitude more efficient.

\noindent
The main contributions of this paper are as follows:
\squishlist
\item We define the novel measure of \emph{Network Inefficiency} and define the problem of finding the \emph{Minimum Inefficiency Subgraph} (\mis), which we prove to be \NPhard. We characterize our measure w.r.t. other existing measures.

 \item We devise a \emph{greedy relaxing algorithm} to approximate the \mis. Our experiments show that in almost all sets of experiments, our greedy relaxing algorithm produces solutions no worse than those produced by an exhaustive search, while at the same time being orders of magnitude faster.

\item We empirically confirm that the \mis\ is a selective connector: i.e., tolerant to outliers and able to detect multiple communities. Besides, the selective connectors produced by our method are smaller, denser, and include vertices that have higher centrality  than the ones produced by the state-of-the-art methods.

\item We show interesting case studies in a variety of application domains (such as human brain, cancer, food networks, and social networks), confirming the quality of our proposal.
\squishend
The rest of the paper is organized as follows. In the next section, we briefly review related prior work and we provide an empirical comparison aimed at highlighting the different characteristics of the connectors produced by our method and state-of-the-art methods.
In Section \ref{sec:problem}, we formally define our problem and prove its hardness in Section~\ref{sec:algorithms}, along with our algorithmic proposals.
Section~\ref{sec:experiments} presents our experimental evaluation and Section~\ref{sec:casestudies} some selected case studies.
Section~\ref{sec:conclusions} concludes this work.

%% file: sections/related.tex
\begin{figure*}[t!]
\centering
\begin{tabular}{cccccc}
\hspace{-4mm}  \includegraphics[width = .15\textwidth]{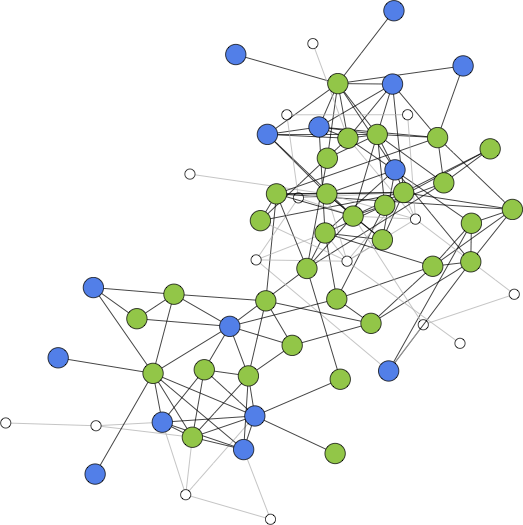} &
\hspace{-2mm} \includegraphics[width = .15\textwidth]{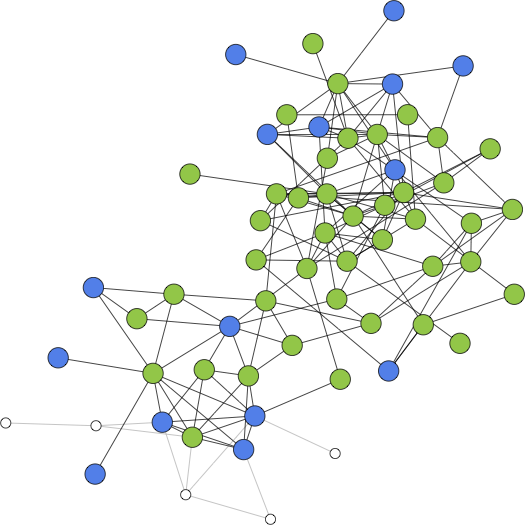} &
\hspace{-2mm} \includegraphics[width = .15\textwidth]{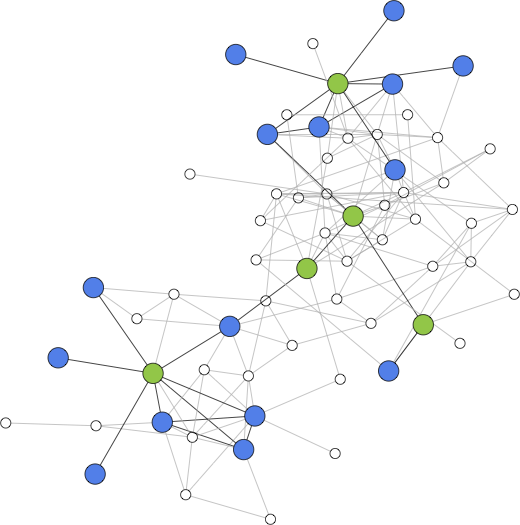} &
\hspace{-2mm} \includegraphics[width = .15\textwidth]{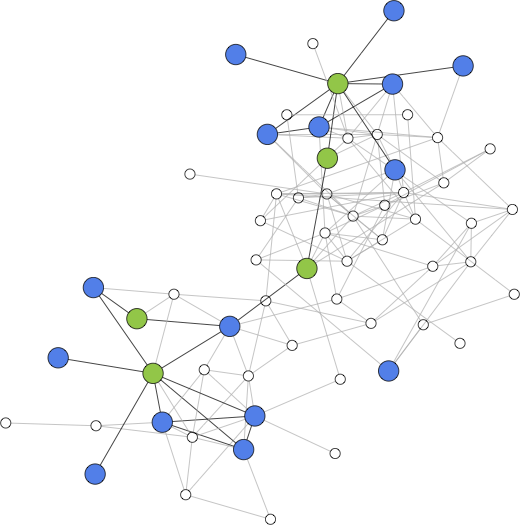} &
\hspace{-2mm} \includegraphics[width = .15\textwidth]{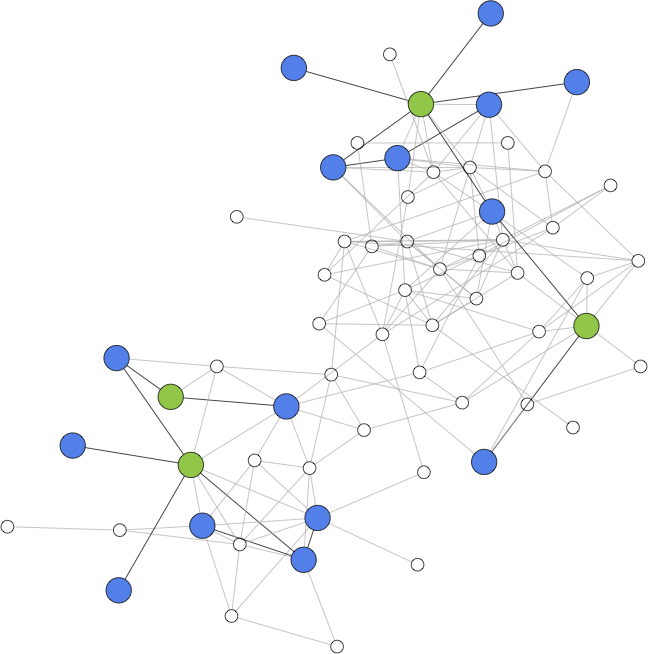} &
\hspace{-2mm} \includegraphics[width = .15\textwidth]{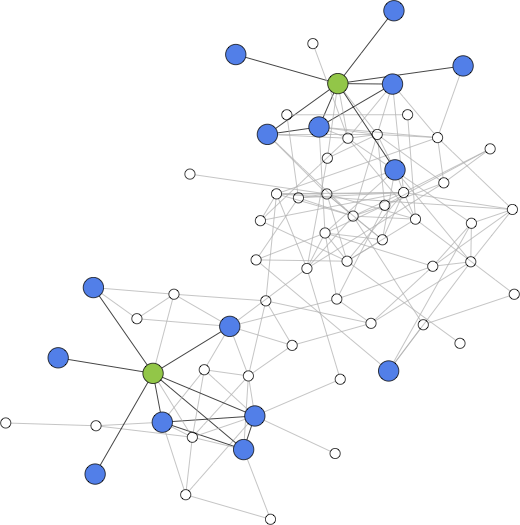} \\
\multicolumn{6}{c}{\vspace{-4mm}}\\
\hspace{-4mm} \cps & \hspace{-2mm} \ctp & \hspace{-2mm} \mwc & \hspace{-2mm} \bh & \hspace{-2mm} \mdl & \hspace{-2mm} \mis
\end{tabular}
\caption{Comparison on the \textsf{Dolphins} social network: query vertices are in blue, added vertices are in green. \label{fig:dolphins}}
\vspace{3mm}
\end{figure*}

In this section, we briefly cover the state of art with respect to algorithmic methods for constructing connectors that connect a set of query vertices without leaving any behind, as well as methods that are more selective in this process and relax the required connectedness.
We also provide a first empirical comparison between these methods and our proposal on two small networks.

\subsection{Connectors and selectors}

Many authors have adopted random-walk-based approaches to the problem of finding vertices related to a given seed of vertices; this is the basic idea of Personalized PageRank \cite{PPR1,PPR2}.
Spielman and Teng propose methods that start with
a seed  and sort all other vertices by their degree-normalized
PageRank with respect to the seed \cite{Spielman}. Andersen and Lang \cite{Andersen1} and Andersen et al. \cite{Andersen2} build on these methods to formulate an algorithm for detecting overlapping
communities in networks. Kloumann and Kleinberg \cite{Kloumann} provide a systematic evaluation of different methods for \emph{seed set expansion} on graphs with known community structure, assuming that the seed set $Q$ is made of vertices belonging to the same community.

Faloutsos et al.~\cite{connect} address the problem of finding a subgraph that
connects two query vertices ($|Q| = 2$) and contains at most $b$ other vertices, optimizing a measure of proximity based on \emph{electrical-current flows}.
Tong and Faloutsos \cite{CenterpieceKDD06} extend \cite{connect} by introducing the concept of \emph{Center-piece Subgraph} dealing with query sets of any size, but again having a budget $b$ of additional vertices.
Koren et al.~\cite{KorenTKDD07} redefine proximity using the notion of \emph{cycle-free effective conductance}  and propose a branch and bound algorithm.
All the approaches described above require several parameters: common to all is the size of the required solution, plus all the usual parameters of PageRank methods, e.g., the jumpback probability, or the number of iterations.

Sozio and Gionis~\cite{SozioKDD10} define the (parameter-free) optimization problem of finding a connected subgraph containing $Q$ and maximizing the minimum degree.  They propose an efficient
algorithm; however, their algorithm tends to return extremely large solutions 
(it should be noted that for the same query $Q$ many different optimal solutions of different sizes exist).
To circumnavigate this drawback they also study a constrained version of
their problem, with an upper bound on the size of the output community. In this case, the problem becomes \NPhard, and they propose a heuristic where the quality of the solution produced can be arbitrarily far away from the optimal value of a solution to the unconstrained problem.

Ruchansky et al. \cite{ruchansky2015minimum} introduce the parameter-free problem of extracting the \emph{Minimum Wiener Connector}, that is the connected subgraph containing $Q$ which minimizes the pairwise sum of shortest-path distances among its vertices. The Minimum Wiener Connector adheres to the parsimonious vertex addition principle, it is typically small, dense, and contains vertices with high betweenness centrality. However, being a connected subgraph is neither tolerant to outliers nor able to expose multiple communities.

Two recent approaches allow disconnected solutions, although very different in spirit from our approach.
Akoglu et al. \cite{akoglu2013mining} study the problem of finding \emph{pathways}, i.e., connection subgraphs for a large query set $Q$, in terms of the Minimum Description Length (MDL) principle.  According to MDL, a pathway is simple when only a few bits are needed to relay which edges should be followed to visit all vertices in $Q$. Their proposal can detect multiple communities and outliers, but it does not follow the parsimony principle as it might add vertices that do not bring any advantage in terms of cohesiveness. 
A major difference, however, is that their pathways are sets of edges identifying trees.
This way their solution drops the induced-subgraph assumption, and, as such, lacks generality.

Given a graph $G$ and a query set $Q$, Gionis et al. \cite{gionisbump} study the problem of finding a connected subgraph of $G$ that has more vertices that belong to $Q$ than vertices that do not.  For a candidate solution $S$ that has $p$ vertices from $Q$ and $r$ not in $Q$, they define the \emph{discrepancy} of $S$ as a linear combination of $p$ and $r$, and study the problem of maximizing discrepancy . They show that the problem is \NPhard\ and develop efficient heuristic algorithms. The maximum discrepancy subgraph is tolerant to outliers (as it is allowed to disregard part of $Q$), but it cannot detect multiple communities (as the solution is one, and only one, connected component).


\subsection{Empirical comparison with prior art}\label{subsec:first_comparison}

\begin{figure*}[t!]
\vspace{-2mm}
\centering
\begin{tabular}{cccccc}
\hspace{-2mm} \includegraphics[width = .14\textwidth]{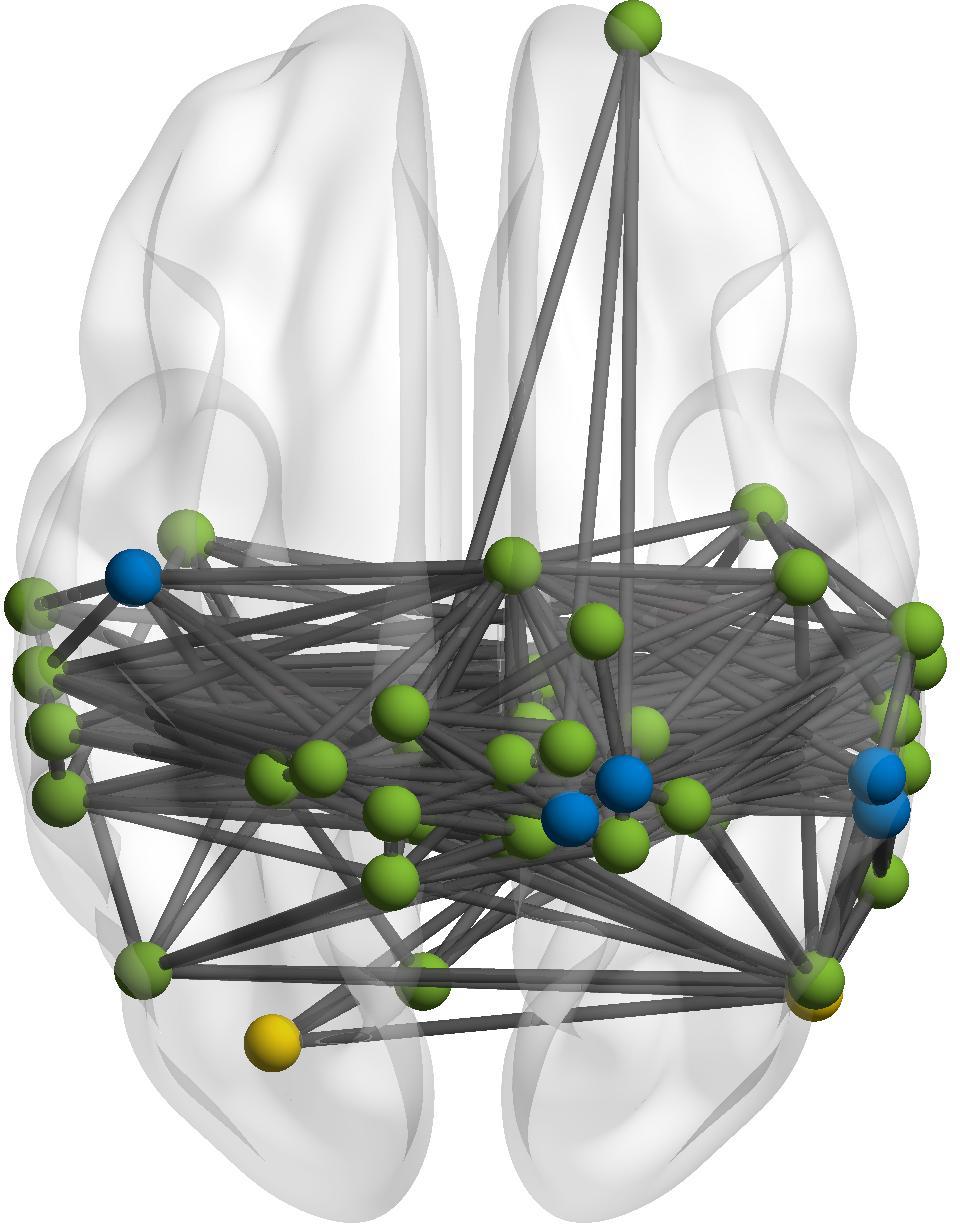}	&
\hspace{-2mm} \includegraphics[width = .14\textwidth]{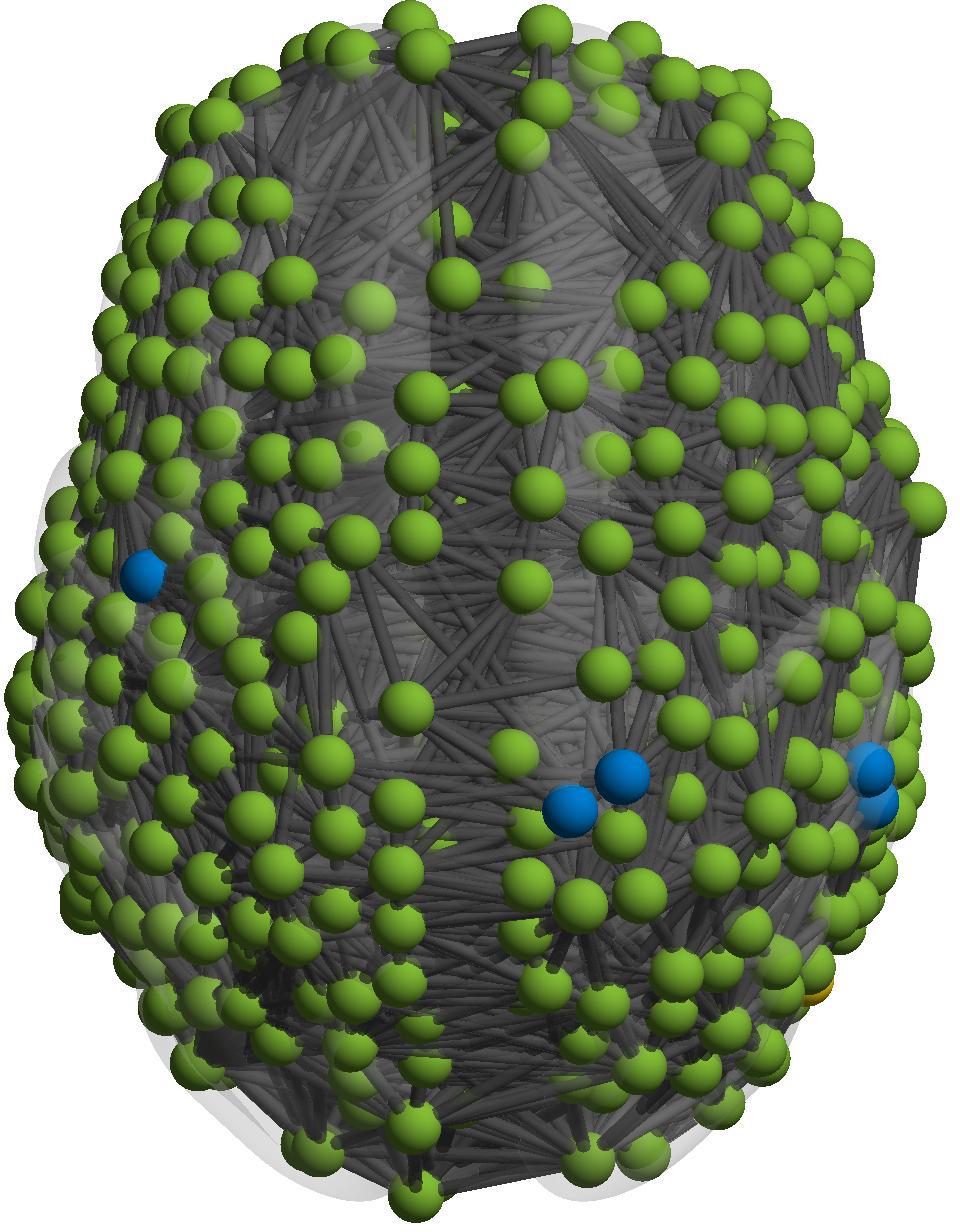}		&
\hspace{-2mm} \includegraphics[width = .14\textwidth]{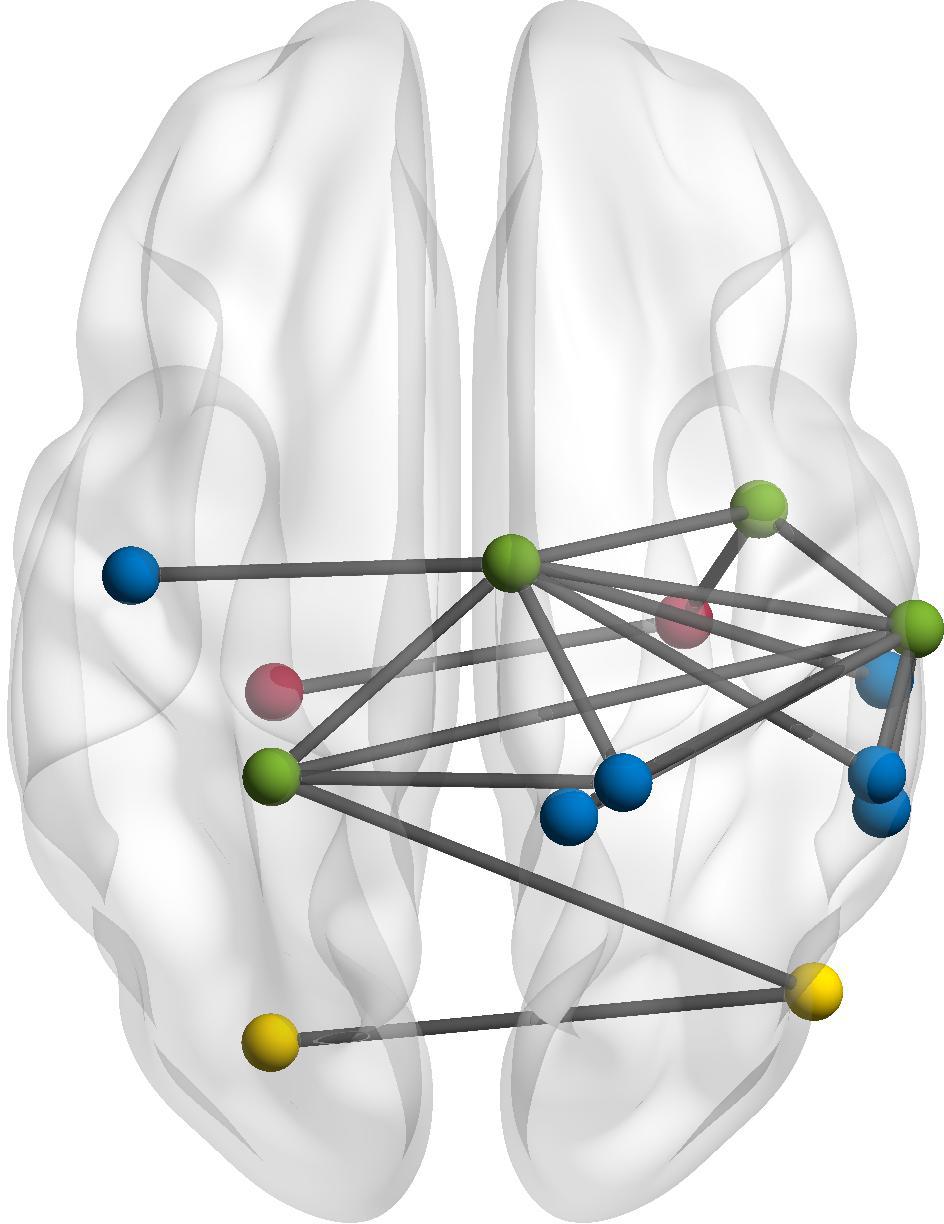}		&
\hspace{-2mm} \includegraphics[width = .14\textwidth]{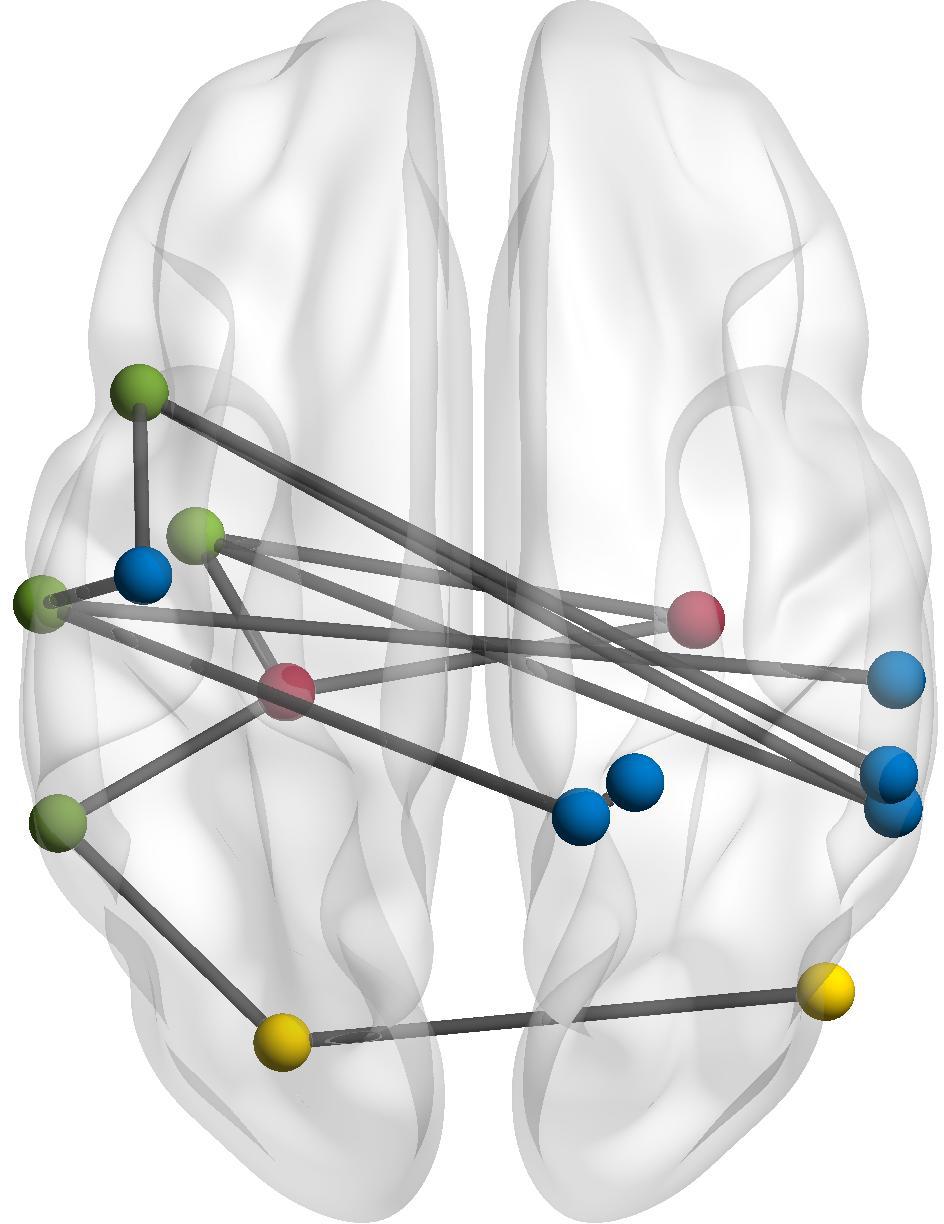}			&
\hspace{-2mm} \includegraphics[width = .14\textwidth]{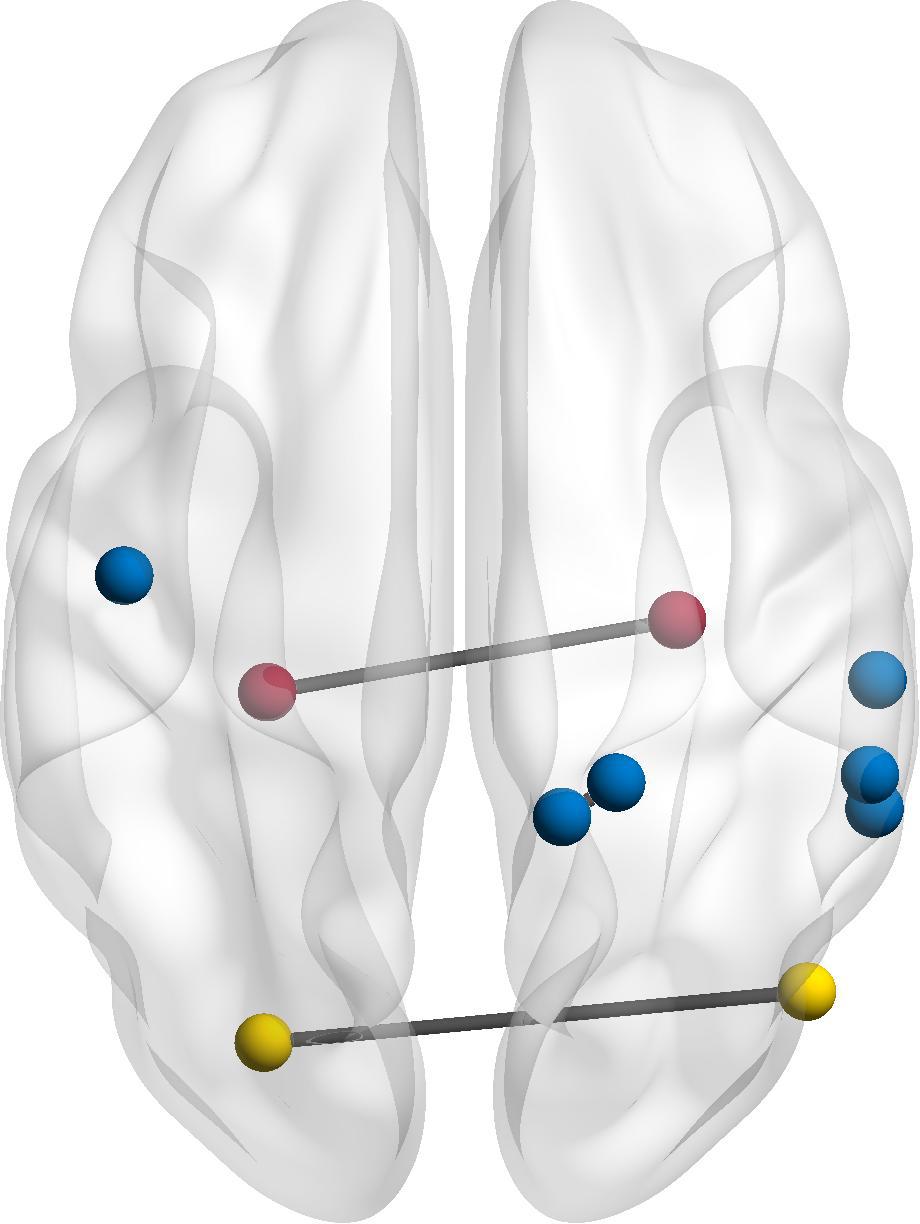}			&
\hspace{-2mm} \includegraphics[width = .14\textwidth]{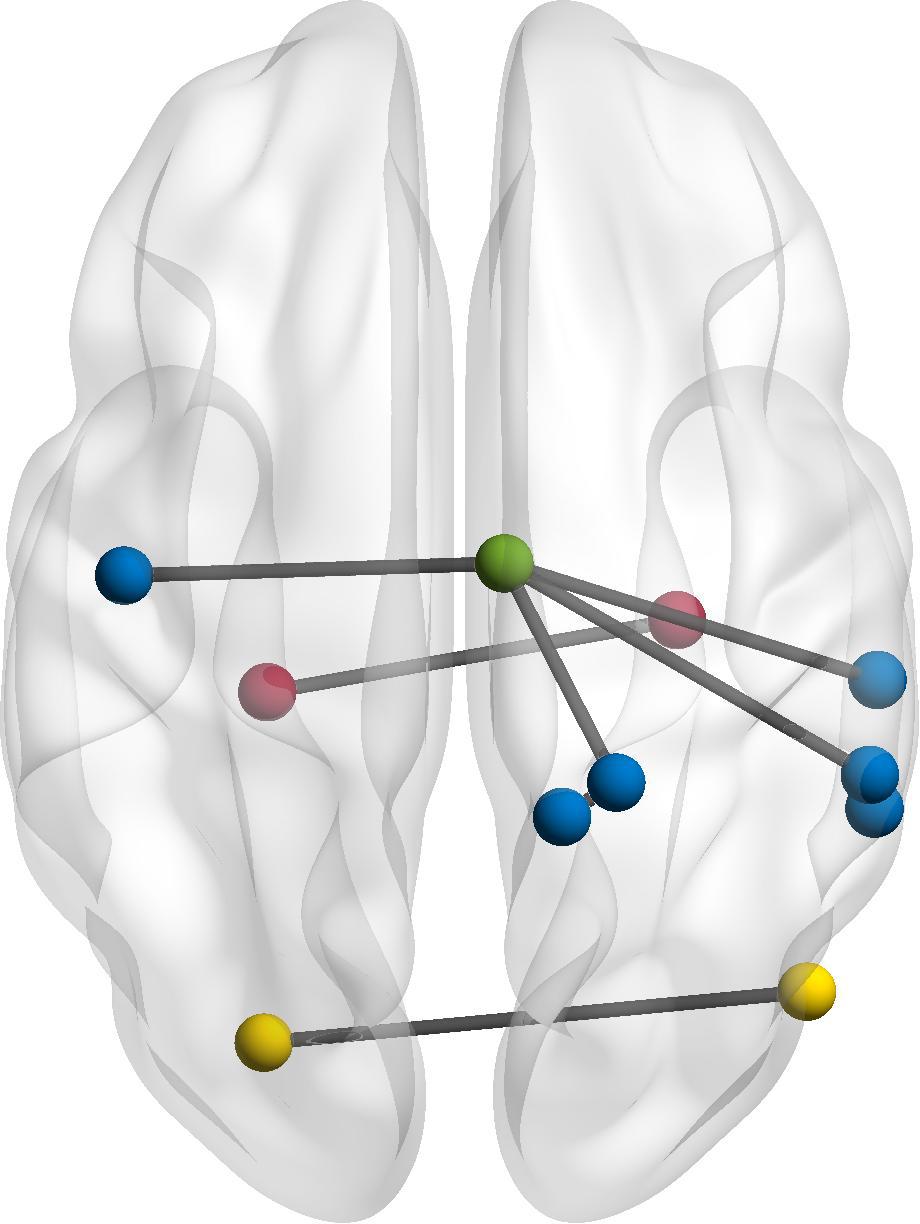}		\\
\hspace{-2mm} \cps & \hspace{-2mm} \ctp & \hspace{-2mm} \mwc & \hspace{-2mm} \bh & \hspace{-2mm} \mdl & \hspace{-2mm} \mis	\\
\end{tabular}
\vspace{-2mm}
\caption{Comparison on a cortical connectivity network. Query vertices are colored w.r.t. their known functionalities: memory and motor function (blue vertices), emotions (yellow vertices), visual processing (red vertices). The green vertices are the ones added to produce the solution. The images were produced using BrainNet Viewer software~\cite{xia2013brainnet}.   \label{fig:brain} }
\vspace{4mm}
\end{figure*}

In the rest of this section, we present two concrete examples on small graphs to highlight the differences in the types of connectors produced by our proposal and these methods in the literature. 
Deeper empirical analysis, more comparisons and case studies will be discussed later in Section \ref{sec:experiments}.

\noindent
The algorithms we consider in the comparison are as follows:\footnote{We do not compare directly with \emph{prize-collecting Steiner tree} because (1) the need of setting vertex weights makes direct and fair comparison difficult, and (2) the algorithm for prize-collecting Steiner tree is at the basis of the heuristic of  \bh, which is covered in our comparison.}

 \squishlist

 \item The random-walk based algorithm for \emph{centerpiece subgraph} (denoted \textbf{\cps}) with restart probability set to 0.9 as recommended in the original paper \cite{CenterpieceKDD06}.

 \item The iterative peeling algorithm for the so-called \emph{``cocktail party'' } (denoted \textbf{\ctp}), in its parameter-free version~\cite{SozioKDD10}.

\item The parameter-free algorithm for the Minimum Wiener Connector (Algorithm 1 in \cite{ruchansky2015minimum}, denoted \textbf{\mwc}).

 \item For the \emph{local discrepancy maximization} (denoted \textbf{\bh}) \cite{gionisbump} there are two algorithmic steps: ($i$) expand the query set and ($ii$) search within the expanded graph. For ($i$) we use the \texttt{AdaptiveExpansion} algorithm, and for ($ii$) we use \texttt{Smart-ST} because they have the best efficiency-accuracy tradeoff \cite{gionisbump}. All the parameters are set with their default value.

\item For the MDL-based approach of \cite{akoglu2013mining} (denoted \textbf{\mdl}), we use the \texttt{Minimum-Arborescence} algorithm because it has the best efficiency-accuracy tradeoff.

\item Finally, the greedy algorithm for \emph{minimum inefficiency subgraph} (denoted \textbf{\mis}) that we introduce later in Section \ref{sec:algorithms}.
\squishend

\spara{Dolphins social network.} Figure \ref{fig:dolphins} reports an example on the famous Dolphins toy-graph\footnote{\url{https://networkdata.ics.uci.edu/data.php?id=6}}: the query vertices are in blue, the vertices added to produce the solution are in green. The query vertices are selected in such a way that there are two clear communities among the vertices in $Q$ and one outlier vertex.

As it is often the case, \cps\ and \ctp\ return a very large solution, while 
\mwc\ produces a much slimmer connector.  As connectedness is still a requirement for \mwc, it is, of course, not able to detect the two communities nor the outlier.
The next three methods (right-half of Figure \ref{fig:dolphins}) allow disconnected solutions.
As discussed above, \bh\ only returns one connected component: thus it can deal with outliers (as it does in the example in Figure \ref{fig:dolphins}) but it cannot return multiple communities.
Regardless of the fact that it aims at producing slim connectors (pathways), \mdl\ adds more green vertices than are strictly needed to connect $Q$. Although in principle it is tolerant to outliers, in this example it does not detect the outlier and pays the price of a bridging green vertex to connect it. Instead, our \mis\ only adds one green vertex for each of the two communities and does not connect the outlier.

\spara{Human connectome.}
Recently there has been a surge of interest in modeling the brain as a graph, with complex topological and functional properties as graph problems \cite{sporns2005human}.
 For our case study, we use a publicly available co-activation dataset\footnote{\url{https://sites.google.com/site/bctnet/datasets}} that was originally described by Crossley et al in ~\cite{crossley2013cognitive}.  The graph contains 638 vertices each of which corresponds to a (similarly sized) cortical area of the human brain. The 18625 links represent functional associations among the cortical areas. In addition, each vertex is associated with location coordinates. In this context, an interesting question is: \emph{given a set of cortical areas in the brain, what are the functional relationships among them?}

As query vertices, we select cortical areas with different known functions. In particular, in Figure~\ref{fig:brain}, $Q$ is the union of the blue, red, and yellow vertices, while the green vertices are the added ones. We used the Talairach Client\footnote{\url{http://www.talairach.org/}} to map each vertex to a Brodmann area using its coordinates.  Brodmann areas are 52 areas of the brain that have been associated with various brain functions through fMRI analysis.\footnote{\url{http://www.fmriconsulting.com/brodmann/index.html}} 
By analyzing the functional associations of the Brodmann area of each vertex, we find that the blue vertices are all involved in memory and motor function (some more in memory and some more in motion), the yellow vertices correspond to Brodmann areas related to emotion, and the red vertices to visual processing.

The \mis\ uncovers these functional similarities in a way that is easy to visualize and interpret: it only adds one vertex that acts as a hub for the blue vertices, without connecting them to the yellow or red vertices. The added vertex corresponds to Brodmann area 6 which contains the premotor cortex and is associated with both complex motor and memory functions.
In contrast, \cps, \ctp, and \mwc\ add many more vertices in order to connect the whole query set.
Although \bh\ could leave out query vertices from the solution, in this specific case it returns a completely connected structure. \mdl\ correctly detects the yellow and red
substructures, but it misses the blue structure, considering the blue vertices as outliers that are too far away to be worth connecting.  This behavior can be explained by the
fact that {\mdl} explicitly penalizes vertices of high degree and that the brain network we analyse is very dense. In fact, for the majority of query sets we experimented with, {\mdl} returned only the edges induced by the query vertices themselves.

%% file: sections/problem.tex
We start by introducing and characterizing  \emph{Network Inefficiency} for a general directed, possibly weighted, graph $G = (V,E)$. We denote $d_{G}(v,u)$ the shortest-path distance between $v$ and $u$ in $G$.
\begin{mydefinition}[Network Inefficiency]
Given a graph $G = (V,E)$ we define its inefficiency as
$$\mathcal{I}(G) = \sum_{\substack{u,v \in V \\ u\neq v}}{1-\frac{1}{d_{G}(v,u)}}.$$
\end{mydefinition}
By taking  the reciprocal of the shortest-path distance, we can smoothly handle disconnected vertices, i.e., $d_{G}(v,u) = \infty$. The same intuition is at the basis of
\emph{network efficiency} \cite{latora2001efficient} (which we already discussed in Section 1), as well as \emph{harmonic centrality}.

    \begin{mydefinition}[Harmonic Centrality] The harmonic centrality of a vertex $u$ in a graph $G = (V,E)$ is defined as
    $$c(u)=\sum_{v\in V}\frac{1}{d_G(v,u)}.$$
    \end{mydefinition}

Next, we show the connection between harmonic centrality, network efficiency, and  network inefficiency. Let $C(G)$ denote the sum of harmonic centrality for all the vertices, i.e.,  $C(G) = \sum_{u \in V} c(u)$ and let $|V| = n$. We note that $C(G)$ ranges in $[0,n(n-1)]$, being 0 \emph{iff} $|E| = 0$ and $n(n-1)$ \emph{iff} $|E| = n(n-1)$ ($G$ is a clique). Then:
    \vspace{1mm}
    \begin{eqnarray*}
      \eff{G} &=& C(G) / (n(n-1)) \\
      \mathcal{I}(G) &=& n(n-1) - C(G)
    \end{eqnarray*}
    \vspace{1mm}
Therefore, efficiency normalizes $C(G)$ by $n(n-1)$ (the maximum possible number of edges in a network of $n$ vertices), thus ranging in $[0,1]$, while inefficiency takes the difference between $n(n-1)$ and $C(G)$, thus ranging in $[0,n(n-1)]$.
Of course, the higher the total harmonic centrality $C(G)$, the more cohesive the graph, giving a higher efficiency and a lower inefficiency.

So far, it is not apparent why we need to introduce network inefficiency, instead of simply relying on the well-established concept of network efficiency.
In Section~\ref{sec:intro}, we already provided a hint: the next examples explain why network efficiency, by not adhering to the parsimonious vertex addition, is not suited for our purpose of extracting selective connectors.

\begin{myexample}
Consider a query set $Q = \{v_1,v_2,v_3\}$ such that the three query vertices are disconnected.
In this case, $C(G[Q]) = 0$, and thus $\eff{G[Q]} = 0$ and  $\mathcal{I}(G[Q]) = 6$.
Consider now $S = Q \cup \{v_4\}$, where a new vertex $v_4$, not connected to $Q$, is added. Again, $C(G[S]) =  \eff{G[S]} = 0$, but $\mathcal{I}(G[S]) = 12$: according to network efficiency $Q$ and $S$ are equivalent.
Instead, by adding a totally unrelated vertex to $Q$, network inefficiency gets worse (larger), as desirable.

Consider instead of adding $v_4$, adding to $Q$  a big clique $T$ with $|T| = 100$, which is disconnected from $Q$ and let $S = Q \cup T$.
In this case, $C(G[S]) = 9900$ and $\eff{G[S]} = 0.942$.
By adding a totally disconnected clique, the network efficiency has gone from minimal to almost maximal.\footnote{This problem is also known in the literature as \emph{free rider effect} \cite{wu2015robust}.}
Instead, network inefficiency gets much worse when adding $T$: in fact, it goes from $\mathcal{I}(G[Q]) = 6$ to $\mathcal{I}(G[S]) = 606$.
\end{myexample}

Therefore, towards our aim of extracting selective connectors, in this paper we study the problem of \emph{Minimum Inefficiency Subgraph} which we formally introduce next.

\spara{Problem statement.}
When introducing network inefficiency above, for the sake of generality, we considered a directed graph.
From now on, when studying the \emph{Minimum Inefficiency Subgraph} problem, we consider a simple, undirected, unweighted graph $G = (V,E)$. Given a set of vertices $S \subseteq V$, let $G[S]$ be the subgraph of $G$ induced by $S$: $G[S] = (S, E[S])$, where $E[S] = \{(u,v) \in E \mid u \in S, v \in S\}$.

\begin{problem}[\textsc{Min-Inefficiency-Subgraph}]\label{prob:mis}
Given an undirected graph $G = (V,E)$ and a query set $Q \subseteq V$, find
$$H^* = \argmin_{G[S] : Q \subseteq S \subseteq V} I(G[S]).$$
\end{problem}

%% file: sections/algorithms.tex
In this section, we first establish the complexity of \textsc{Min-Inefficiency-Subgraph} and then present our algorithms.

\subsection{Hardness}
\input{sections/hardness.tex}

\subsection{Greedy relaxing algorithm}
Given that finding the \emph{Minimum Inefficiency Subgraph} exactly is hard, we now search for an algorithm that approximates it accurately and efficiently.  One approach would be
to start from the whole graph and search for the subgraph that minimizes $\mathcal{I()}$.  Not only is this approach costly, but it is also highly unnecessary.  

Recall from Section
\ref{sec:problem} that the Network Inefficiency of a subgraph $S$ can be written as a difference of two terms $\mathcal{I}(S)=|S|(|S|-1)-C(S)$. Hence, when considering a candidate subgraph $S$ as a {\mis}, we can think of the cost as a balance of the two terms $|S|(|S|-1)$ and $C(S)$.  If the query vertices are far apart, then connecting them will require many vertices which will make the left-hand term grow faster than the right-hand term, and, as a result, $\mathcal{I}(Q)$ will be smaller.  
This shows that the cost of not connecting the query vertices at all, i.e.,  $\mathcal{I}(Q)$, acts as an upper bound tolerance on the candidate {\mis}, and implies that our search for {\mis} need not explore the whole graph, and can remain fairly local.

Motivated by this observation, we follow an approach based on first finding a connector, i.e. a subgraph $H$ of $G$ that connects all of $Q$, and then \emph{relaxing} the connectedness requirement by iteratively removing non-query vertices that incur a large inefficiency cost.  In the choice of initial connector there are two properties we desire: (1) it should contain a superset of vertices that are parsimonious in the sense that they are cohesive with $Q$, and (2) it should be small to prompt an efficient algorithm. Given the resemblance of the objective function based on shortest-path distances, and the fact of being parameter free, the \emph{Minimum Wiener Connector} \cite{ruchansky2015minimum} (\mwc) is the most natural choice.

We recall that the {\mwc} is the subgraph $H$ of $G$ that connects all of $Q$ and minimizes the sum of pairwise shortest-path distances among its vertices: i.e.,
\vspace{3mm}
$$
H = \argmin_{G[S] : Q \subseteq S \subseteq V} \sum_{\{u,v\} \in S} d_{G[S]}(u,v)
$$
\vspace{3mm}

Algorithm~\ref{algo:greedy} provides the detailed pseudocode of the proposed \emph{greedy relaxing algorithm} for minimum inefficiency subgraph, that we denote \greedy. Our
proposed algorithm takes as input a graph $G$ and a set of query vertices $Q$, and starts by constructing the \mwc, as the candidate connector $G[S]$ (line 1). Next, the algorithm
iteratively removes from $S$ the non-query vertex whose removal results in the smallest value of network inefficiency (lines 3--11), until all non-query vertices have been removed. Among all the intermediate subgraphs created during the greedy relaxation process, the subgraph with the minimum value of network inefficiency is returned (line 11).

\vspace{3mm}
\begin{algorithm} [h!]
\caption{ \textbf{--} \greedy: Greedy Relaxing Algorithm for \mis}
	\begin{algorithmic}[1]
	\Statex \textbf{Input:} Graph $G=(V,E)$, query vertices $Q \subseteq V$
	\Statex \textbf{Output} Selective connector $G[S]$ s.t. $Q \subseteq S \subseteq V$
	\State $G[S]\leftarrow {\mwc}(Q)$ (Algorithm 1 in \cite{ruchansky2015minimum})
\State $i \leftarrow 0$	
\While{$|S|>|Q|$}
        \State $G_i\leftarrow G[S]$
        \State $i\gets i + 1$
		\For {$u\in \{S\setminus Q\}$}
			\State $c_u \leftarrow \mathcal{I}(G[S\setminus \{u\}])$
		\EndFor
		\State $v\leftarrow\argmin_{u} c_u$
        \State $S \leftarrow S \setminus \{v\}$
	\EndWhile
	\State \textbf{Return} $\argmin_{j \in [0,i]} \mathcal{I}(G_j)$
	\end{algorithmic}
\label{algo:greedy}
\end{algorithm}

Parsimonious vertex addition is guaranteed by starting with ${\mwc}(Q)$ and then ``relaxing'' it.
At one extreme, if a cohesive subgraph exists that connects all the vertices in $Q$, then this would be captured by $G_0 = {\mwc}(Q)$. At the other extreme, if no good connection exists, the minimum inefficiency is obtained by  $Q$ itself without adding any vertex (captured by $G_j$ with $j = |S\setminus Q|$). 
As far as the other two design requirements of outlier tolerance and identification of multiple communities, they are both satisfied by the fact that the solution subgraph is not necessarily connected: singleton solution vertices may be interpreted as outliers, while every (non-singleton) connected component can be viewed as corresponding to a different community. 
Moreover, our algorithm complies with the induced-subgraph assumption, thus being able to output \emph{general subgraphs}.

\spara{Computational complexity.} Typically, the most time-consuming step of $\ouralg$, is the extraction of the \mwc\ (line 1), which may be computed in time
$\widetilde{\mathcal{O}}(|Q| \cdot |E(G)|)$~\cite{ruchansky2015minimum}. In fact, the subgraphs returned are typically not much larger than the
query set itself, and the remainder of the algorithm only operates on the subgraph induced by the solution.
Let $S$ denote the set of vertices corresponding to \mwc. We analyze the cost of steps 2--12 in terms of
$\tilde{n} = |S|$ and $\tilde{m} = |E[S]|$ (the number of vertices and edges of the subgraph induced by $S$, respectively).
Keep in mind that, unless $|Q|$ contains a sizable fraction of the graph, $\tilde{n}$ and $\tilde{m}$ are usually much smaller than $|V|$ and $|E|$.

Each iteration of the \textbf{while} loop performs $|S \setminus Q| \le \tilde n$ iterations of an all-pairs shortest path computation on (a subgraph of) $G[S]$.
All pairs-shortest paths in the unweighted, undirected graph $G_i$ may be computed in time $\mathcal{O}(\tilde{n} \cdot (\tilde{n} + \tilde{m}))$.
Hence, the \textbf{while} loop takes time $\mathcal{O}(\tilde{n}^2 \cdot (\tilde{n} + \tilde{m}))$, and since there are
$|S\setminus Q| \le \tilde n$ iterations, the overall complexity of \ouralg, excluding the time spent on line 1, is $\mathcal{O}(\tilde{n}^3(\tilde{m}+\tilde{n}))$.

\subsection{Baselines}\label{subsec:baselines}
In our empirical comparison (Section \ref{sec:experiments}), besides comparing with the state-of-the-art methods already listed in Section \ref{subsec:first_comparison}, we 
justify the appropriateness of our choices. In particular, we need to show that $(i)$ \mwc\ is a good choice as starting connector, and $(ii)$ greedily relaxing the connector, while giving us an efficient search, does not lose much in quality with respect to an exhaustive search.

For the first point, we will compare against two variants of the greedy relaxing algorithm, which start with different connectors: the \emph{centerpiece subgraph} \cite{CenterpieceKDD06} (we denote this variant \greedycps), and the \emph{cocktail party subgraph}~\cite{SozioKDD10} (denoted \greedyctp).

For the second point, we consider an algorithm that starts with the \mwc, but instead of relaxing the connector greedily, it performs an \emph{exhaustive} search by considering the removal of all possible subsets of non-query vertices $S\setminus Q$.  We denote this algorithm \exh.
Since \exh\ explores a number of subgraphs which is exponential in $|S\setminus Q|$, it is
clearly computationally expensive, and becomes unfeasible for large $S$. For this reason, \exh\ cannot be started with \cps\ or \ctp\ as the initial connector to be relaxed, as both \cps\ and \ctp\ typically return a much larger starting subgraph.

%% file: sections/hardness.tex
\begin{theorem}
\textsc{Min-Inefficiency-Subgraph} is \NPhard, and it remains hard even on undirected graphs with diameter 3.
\end{theorem}
\begin{proof}
We show a polynomial-time reduction from 3-SAT. Let $\phi = \bigwedge_{i=1}^m (l_i^1 \vee l_i^2 \vee l_i^3)$ be an
instance of 3-SAT with $m$ clauses, where
$l_i^j$ stands for the $j$th literal in clause $i$,
and all literals in a clause refer to different variables.

Let
$$ M    \triangleq 6 m ^2 + 1, $$
$$ B_1  \triangleq M^2 m \left(m - \frac{1}{2}\right),\quad B_2  \triangleq M m \frac{m - 1}{2}. $$

Given $\phi$, we construct a graph $G = (V, E)$ and a query set $Q \subseteq V$ as follows.
First we introduce the following vertices in $G$ for each clause $C_i = l_i^1 \vee l_i^2 \vee l_i^3$ of $\phi$:
\begin{itemize}
    \item two disjoint sets of vertices of size $M$, namely $A_i = \{a_i^1, \ldots, a_i^M\}$ and $B_i = \{b_i^1, \ldots,
    b_i^M\}$.
 \item a set $S_i = \{s_i^1, \ldots, s_i^7\}$ of 7 new vertices, representing all the assignments of the
    three literals of $C_i$ satisfying $C_i$. 
\end{itemize}

The vertex set of $G$ is $V = \bigcup_{i=1}^m A_i \cup \bigcup_{i=1}^m B_i \cup \bigcup_{i=1}^m S_i$.
The edges of $G$ are the following:
\begin{itemize}
    \item $(a_i^t, a_j^r)$ and $(b_i^t, b_j^r)$ if either $i \neq j$ or $r \neq t$ holds, for each $i, j \in [m]$ and $r, t \in [M]$;
    \item $(a_i^t, s_i^j)$ and $(s_i^j, b_i^t)$  for each $i \in [m], t \in [M], j \in [7]$; 
    \item $(s_i^j, s_{i'}^{j'})$ if $s_i^j$ and $s_{i'}^{j'}$ refer to compatible assignments, for each $i \in [m], t \in [M], j \in [7]$.
          Two assignments are \emph{compatible} if every variable in common variable receives the same truth value in both.
\end{itemize}

\noindent
The query set is $Q = \bigcup_{i=1}^m A_i \cup \bigcup_{i=1}^m B_i$.

Clearly $G$ and $Q$ can be constructed in time $\poly(m) = \poly(|\phi|)$.
Note also that the diameter of $G$ is 3 by construction.
It remains to be shown that the reduction is correct:

\begin{center}
$\phi$ is satisfiable $\Leftrightarrow G$ has a Minimum Inefficiency Subgraph of $Q$ with cost $\le B_1 + B_2$.
\end{center}

If $\phi$ is satisfiable, then our intended solution will contain a path of length two between each element of $A_i$ and
each element of $B_i$, through some element of $S_i$ (representing a partial assignment). Moreover, these partial assignments will be shown to be
extensible to a full satisfying assignment for $\phi$. Details follow.

First observe that for any $T \supseteq Q$, the following inequalities hold:
\begin{itemize}
    \item $d_T(a_i^t, a_j^r) = 1$ and $d_T(b_i^t, b_j^r) = 1$ for each $i, j \in [m]$ and $r, t \in [M]$ where either $i \neq j$ or $r \neq t$;
    \item $d_T(a_i^t, b_i^r) \ge d_G(a_i^t, b_i^r) = 2$ for each $i \in [m]$ and $r, t \in [M]$; equality holds if and only if $T \cap S_i \neq \emptyset$.
    \item $d_T(a_i^t, b_j^r) \ge d_G(a_i^t, b_j^r) = 3$ for each $i \neq j \in [m]$ and $r, t \in [M]$;
    equality holds if and only if $T \cap (S_i \cup S_j) \neq \emptyset$.
    \item $d_T(a_i^t, s_i) = 1$ and $d_T(s_i, b_i) = 1$ for each $i \in [m]$ and $s_i \in S_i \cap T$;
    \item $d_T(a_j^t, s_i) \ge d_G(a_j^t, s_i) = 2$ and $d_T(s_i, b_i) \ge d_G(s_i, b_j^t) = 2$ for each $i \in [m]$ and $s_i \in S_i \cap T$;
    equality holds if and only if $T \cap (S_i \cup S_j) \neq \emptyset$.
    \item $d_T(s_i, s_j) \ge 1$; equality holds if and only if $s_i$ and $s_j$ are compatible assignments.
\end{itemize}

Assume that $\phi$ is satisfiable and pick a satisfying assignment $f$ for $\phi$. For each $i \in [m]$, select the element $s_i \in S_i$ that represents the truth-value assignment
of $f$ on the variables of clause $C_i$ (which by assumption satisfies $C_i$). Let $T = Q \cup \{ t_i \mid i \in [m] \}$. Observe that for all $s_i, s_j \in T \setminus Q$,
assignments $s_i$ and $s_j$ are compatible. Then
$$ \mathcal{I}(G[T]) = \frac{1}{2} M^2 n + \frac{2}{3} M^2 \binom{m}{2} + 2 \cdot \frac{1}{2} M \binom{m}{2} = B_1 + B_2. $$

Conversely, consider any solution $T$. 
If $T \cap S_i = \emptyset$ for some $i$, then $d_T(a_i^t, b_i^r) \ge 3$ for at least $M^2$ pairs, so we must have
$$ \mathcal{I}(G[T]) \ge B_1 + \left(\frac{1}{2} - \frac{1}{3}\right) M^2> B_1 + B_2. $$
Otherwise $\mathcal{I}(G[T]) \ge B_1 + B_2$, with equality if and only if
for all $s_i, s_j \in T \setminus Q$, the assignments $s_i$ and $s_j$ are compatible.

Therefore, we conclude that if $\mathcal{I}(G[T]) \le B_1 + B_2$, then $T$ contains partial assignments for every clause and moreover, these partial assignments are pairwise compatible and hence can be extended to a full satisfying assignment for $\phi$, implying that $\phi$ is satisfiable.\end{proof}

%% file: sections/experiments.tex
In this section, we report our empirical analysis which is structured as follows.
In Section~\ref{subsec:qualityofgreedy}, we study the quality and efficiency of connectors that our greedy relaxing algorithm \greedy\ creates, by comparing with its exhaustive counterpart \exh, and with the greedy variants that start from different connectors (\greedycps\ and \greedyctp).
Then, we analyze the structural features of our proposed selective connector \mis, and we compare it with other methods in the literature which provide different selective connectors, namely \bh\ \cite{gionisbump} and \mdl\ \cite{akoglu2013mining} (Sections~\ref{subsec:qualityofmis} and~\ref{subsec:pco}).
Finally, in Section \ref{subsec:scalability}, we discuss the scalability of our method.

\spara{Datasets.}
We experiment with both synthetic and real-world datasets, from a variety of domains.
The synthetic datasets allow us to control various properties of the graphs and, consequently, the expected outcomes of the selective connector algorithms.
Table~\ref{tab:all-nets} provides a summary of real-world datasets used.
All our datasets, with the exception of \dataset{football}\footnote{\url{http://www-personal.umich.edu/~mejn/netdata/}}, come with auxiliary ground-truth communities information.\footnote{\url{http://nodexlgraphgallery.org/pages/Graph.aspx?graphID=26533}}\footnote{\url{http://socialcomputing.asu.edu/datasets/Flickr}}\footnote{\url{https://snap.stanford.edu/data/\#communities}}

\spara{Query selection.}
We define different query sets by exploiting the pre-existing community structure. In particular, we use three parameters:  the number of query vertices, the number of communities they span, and the minimum number of query vertices that should come from the same community. In more details, given a graph $G=(V,E)$ and a community membership vector $C$, where $C(u)=c_i$ indicates that vertex $u$ participates in community $i$, we generate a query set $Q$ with three parameters: $n$, $m$, and $k$ by taking the following steps:
$(1)$ select a random community $c_i$; $(2)$ select $n$ vertices that belong to $c_i$; $(3)$ select $m$ vertices across $k$ other communities $c_j\neq c_i$.
By this construction, we can cover a range of query types.
For example, setting $m=k$ gives the setting of $n$ vertices from one community and $m$ outliers.

\input{tables/datasets}

\subsection{Comparison with baselines}\label{subsec:qualityofgreedy}
We first report the comparison of \greedy\ against the three baselines that we introduced in Section \ref{subsec:baselines}: \exh, \greedycps, and \greedyctp.
Due to the computational complexity of the three baselines, here we consider only small synthetic graphs using the Community Benchmark generator\footnote{\url{https://sites.google.com/site/santofortunato/inthepress2}} as it allows us to easily explore a large variety of graphs.
Specifically, we generate a graph of $200$ vertices with a maximum degree of $100$, and vary the following parameters in the respective ranges: clustering $[0.1,0.5,0.8]$, mixing $[0.1,0.5,0.8]$, and average degree $[5,10,30]$; other parameters are left as default.
In addition, we vary the query parameters $n\in [0,10]$, $m\in[0,10]$, and $k\in[0,m+1]$ to regulate the query set $Q$, and repeat each setup 10 times.
\begin{table}[t!]
\centering
\caption{Quality comparison (w.r.t. $\mathcal{I(\cdot)}$) on synthetic data.}
\vspace{-2mm}
\small
\begin{tabular}{ c | c | c | c |}
\multicolumn{1}{c}{}   & \multicolumn{1}{c}{\greedycps}      &   \multicolumn{1}{c}{\greedyctp}      &  \multicolumn{1}{c}{\exh}   \\    \cline{2-4}
 {\greedy}  $= $          & $0.70$& $0.68$& $0.99$\\ \cline{2-4}
{\greedy}   $ < $          & $0.18$& $0.21$&$0.00$ \\ \cline{2-4}
 {\greedy}   $>$          &$0.12$ &  $0.11$& $0.01$\\ \cline{2-4}
\end{tabular}
\label{tab:comparison-with-exh}
\end{table}

Table~\ref{tab:comparison-with-exh} reports the aggregate result of these experiments (portion of setups that each condition on the left column applies).
Although \greedycps\ and \greedyctp\ start with a larger connector and thus a larger number of vertices to pick from,
we find that for the majority of the time (68-70\%), initializing with the three different connectors leads to a solution of equal cost.
In particular, we observe that on graphs with high clustering and low mixing (properties exhibited by real-world networks with community structure) \greedy\ initialized with
{\mwc} consistently achieves lower cost solutions (up to 90\% of the runs) than the baselines.
The larger size of the connectors {\cps} and {\ctp}, not only lacks the quality advantages, but also it leads to much larger running time, as reported in
Figure~\ref{fig:basetimes}.
Taken together, these observations illustrate that {\mwc} provides a compact connector that contains a superset of the vertices important for the selective connector, as {\cps} and {\ctp} often do, meanwhile remaining small and cost efficient -- making it a great starting point for {\greedy}.
Finally, when comparing {\greedy} and {\exh}, we see that {\greedy} achieves the same cost as {\exh} in almost every instance.
Since the {\exh} algorithm is computationally demanding, {\greedy} offers an efficient and highly accurate alternative.

 \begin{figure*}[t!]
 \vspace{-2mm}
   \centering
   \begin{tabular}{ccc}
 \hspace{-8mm}\includegraphics[scale=.1]{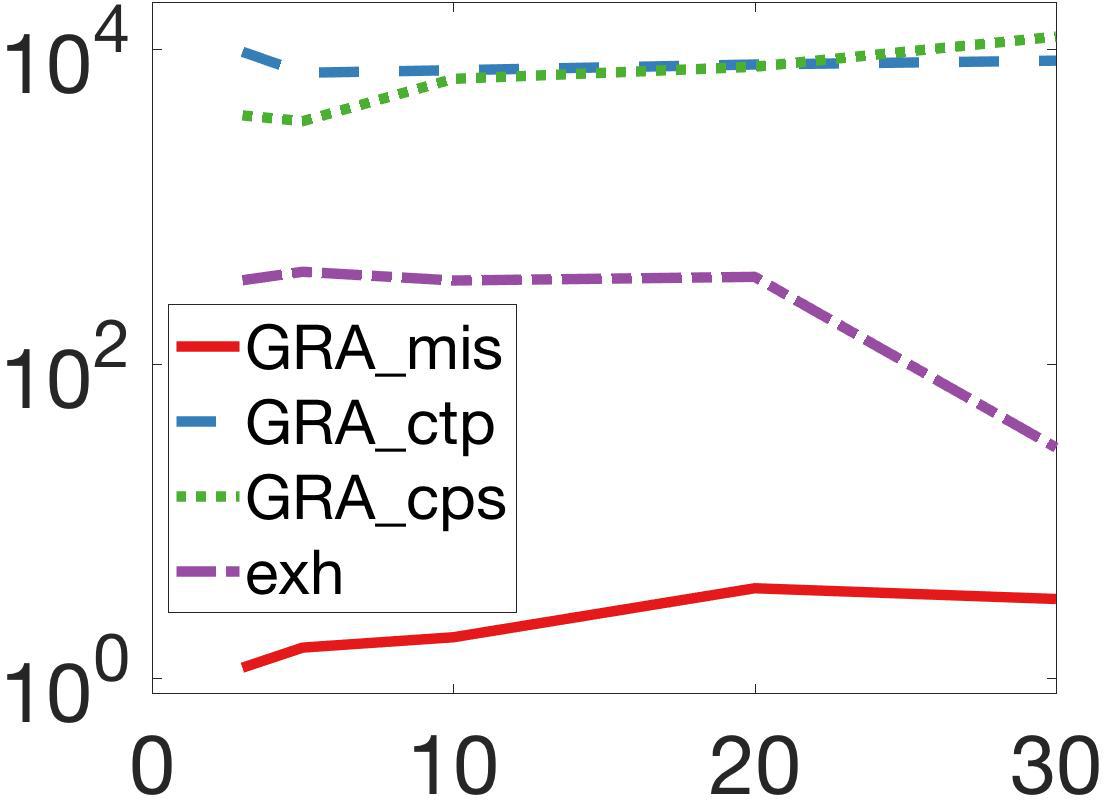} &
 \hspace{-3mm}\includegraphics[scale=.1]{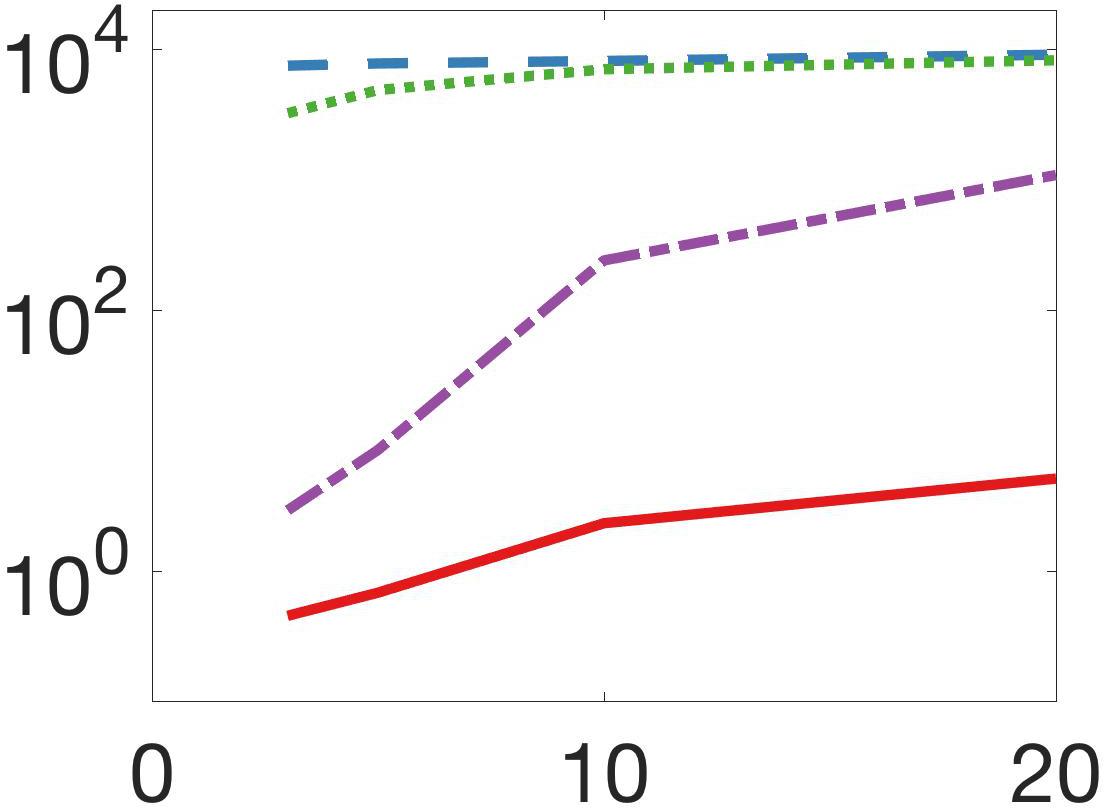} &
\hspace{-3mm} \includegraphics[scale=.1]{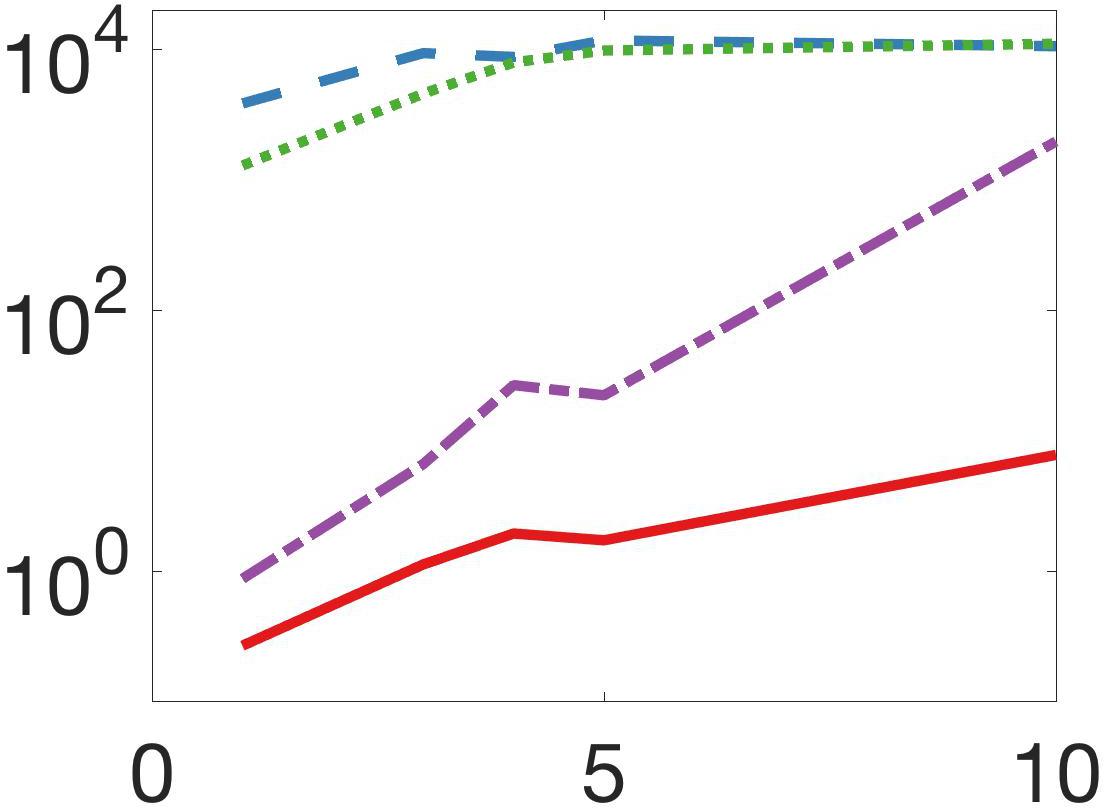}\\
$n$ & $m$ & $k$
    \end{tabular}
    \vspace{-2mm}
 \caption{Runtime (seconds) by varying the parameters $n$, $m$, and $k$, that control the selection of the query set $Q$.\label{fig:basetimes}}
\vspace{2mm}
 \end{figure*}

\subsection{Comparison with state-of-the-art methods}\label{subsec:qualityofmis}

Next, we study the characteristics of \mis\ as a selective connector, comparing with the  state-of-the-art methods that produce a selective connector, namely {\bh} and {\mdl}.
Limited by the runtime of {\mdl}, we study the performance on small synthetic (already described in Section \ref{subsec:qualityofgreedy}) and real networks (\dataset{football} and \dataset{kdd14twitter}). Query set selection is done with parameters  $n=10$, $m=10$, and $k=4$, and averaged over 20 runs. The results, reported in Table~\ref{tab:quality1}, are representative of the behavior observed in other query setups (not reported due to space constraints). 

\input{tables/tablesmall}

In Table~\ref{tab:quality1}, we see that, not surprisingly, {\mis} consistently achieves lower network inefficiency than the other algorithms.
All three methods follow the parsimonious vertex addition principle, returning very compact connectors, usually containing a small number of additional vertices over the query set (recall that in these experiments $|Q|=20$. However, we can see that the solutions returned by \mdl\ have very low density; this is not surprising as the goal of \mdl\ \cite{akoglu2013mining} is to find pathways, not dense substructures.
Similar arguments hold for the centrality measures of the additional vertices.
In all of these measures, \bh\ is closer to \mis, although \mis\ consistently outperforms the other two methods in all the measures.

\subsection{Parsimony, communities, and outliers}\label{subsec:pco}
Next, we move to a thorough evaluation of \mis\ on larger datasets.
We generate query sets by varying $n\in[0,20]$, $m\in [0,20]$, and $k\in[1,m]$.
By varying these parameters, we can check whether the extracted \mis\ exhibits the expected behavior.

The leftmost plot in Figure~\ref{fig:behavior} shows the solution size ($|S|$, on the Y-axis) as a function of the query size ($|Q|$, on the X-axis).
By the parsimonious vertex addition principle, we would like the two to be close to one another: this is exactly the case in all datasets.
For large query sets, regardless of the community spread ($m$), the {\mis} remains fairly small.
This property is crucial for applying \mis\ in real world scenarios where it is important to produce solutions that are easy to visualize and interpret.

The central plot shows the number of connected components ($\#CC$, on the Y-axis) as a function of the number communities in $Q$ ($k$, on the X-axis) .
If the \mis\ is able to detect the communities, we would expect these two quantities to be similar: again as most of the lines lie close to the diagonal, we can conclude that the \mis\ exhibits the expected behavior w.r.t. detection of multiple communities.

Finally, the rightmost plot of Figure~\ref{fig:behavior} shows the number of singleton vertices ($|Q_d|$, on the Y-axis) in the solution as a function of $m$, for the setups where $k=m$.
Recall that when $k=m$, the query set $q$ contains $n$ vertices from one community and $m$ vertices, each from a different community.
The nearly linear trend observed in all datasets shows that the {\mis} successfully identifies and disconnects outlier vertices -- satisfying the desired outlier detection property.

 \begin{figure*}[t!]
   \centering
   \begin{tabular}{ccc}
 \hspace{-6mm}\includegraphics[scale=.22]{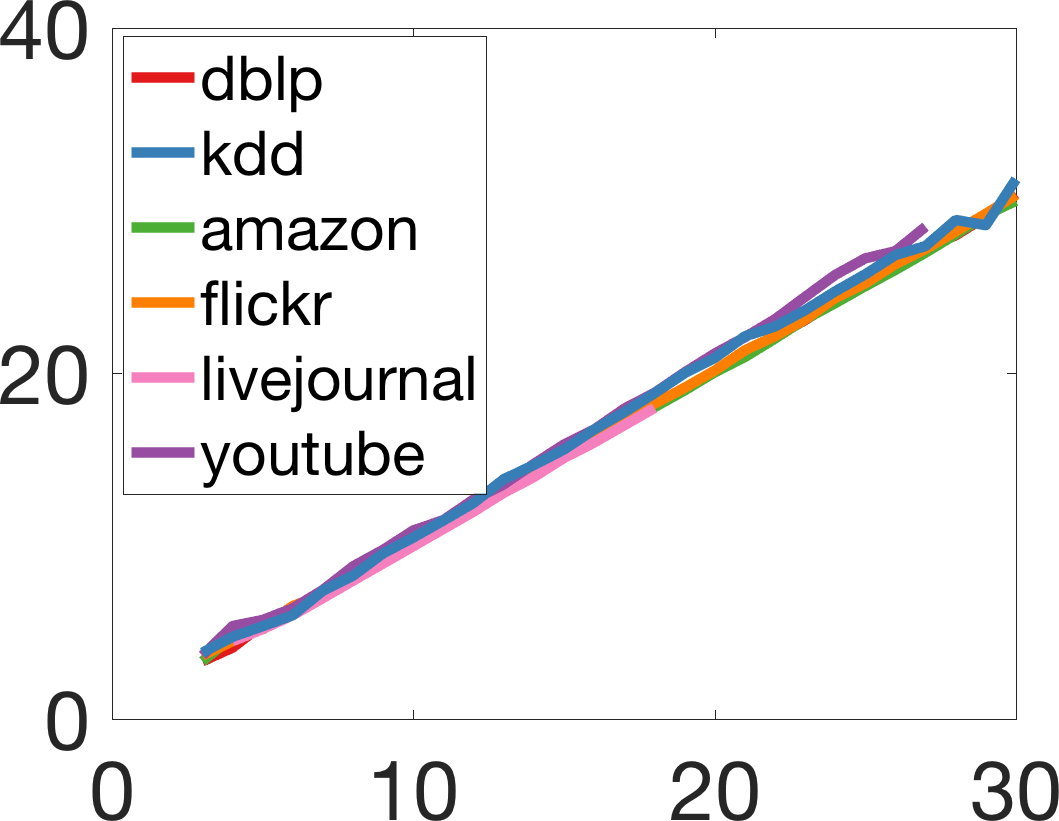} &
 \hspace{-3mm}\includegraphics[scale=.22]{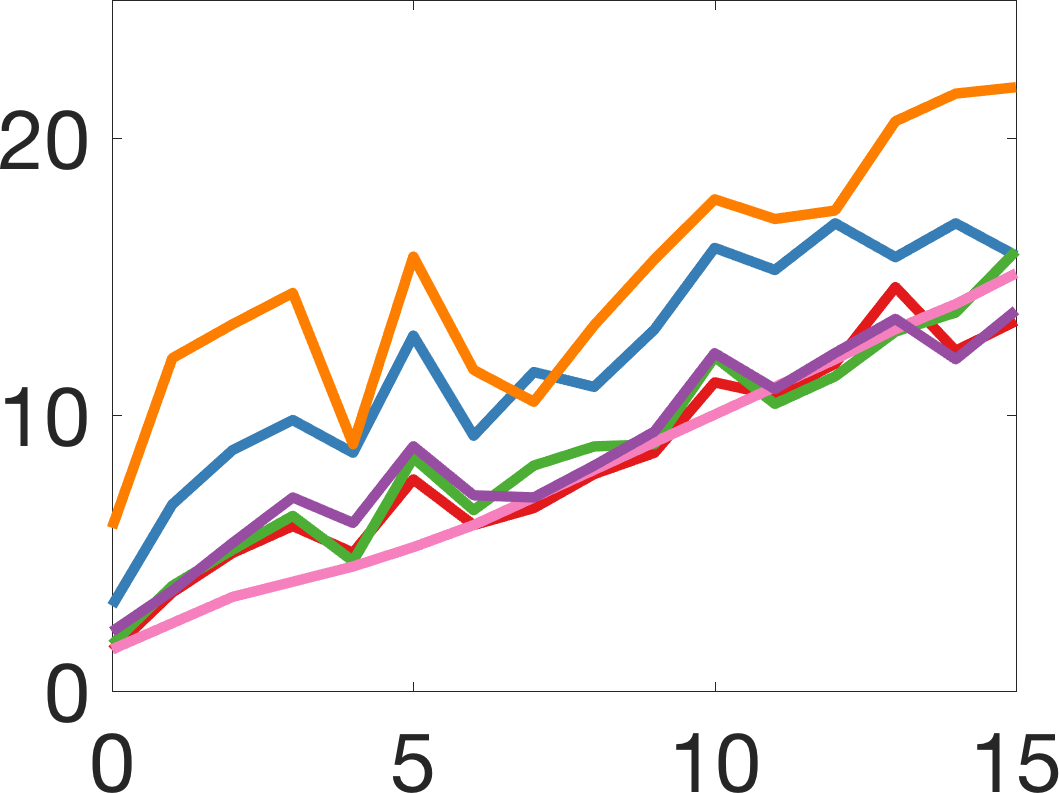} &
\hspace{-3mm} \includegraphics[scale=.22]{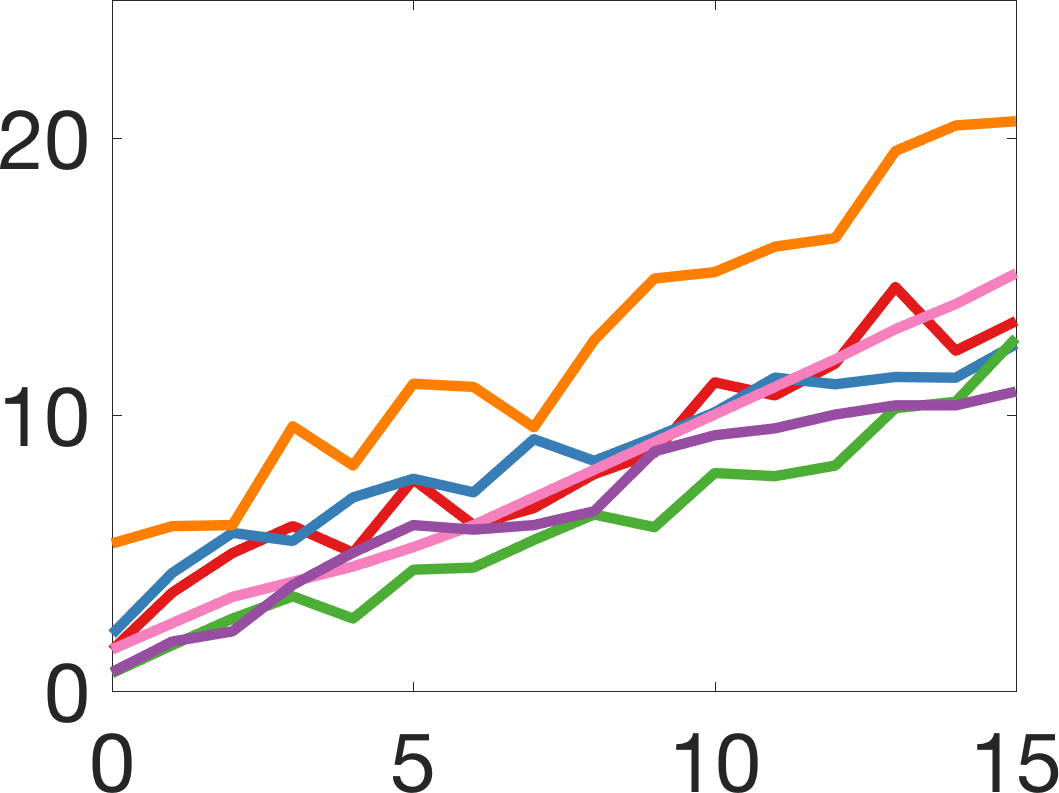}\\
$|Q|$ vs. $|S|$ & $k$ vs. $\#CC$ & $m$ vs. $|Q_d|$
    \end{tabular}
 \caption{Properties of the Minimum Inefficiency Subgraph on the large datasets with a large variety of query setups.\label{fig:behavior}}
\vspace{4mm}
 \end{figure*}

\subsection{Scalability considerations}\label{subsec:scalability}
Our Python implementation of the {\greedy} algorithm can easily run on large graphs. For instance, on \dataset{livejournal} ($|V| = 4M, |E| = 34M$) with $|Q| = 10$, \greedy\ takes less than 10 minutes on an Intel-Xeon CPU E5-2680 \@ 2.70GHz equipped with 264Gb of RAM. Of these 10 minutes, only 7 seconds are taken by the greedy relaxing algorithm, while the rest of the time is spent on computing the initial connector.
In case scaling to even larger graphs is necessary,  the extraction of the \mwc\  can be parallelized
as described
in~\cite{ruchansky2015minimum}. Other techniques can be used to speed up {\greedy} further:
\squishlist
\item Constrain the re-computation of shortest paths to operate only in the connected components that are affected by the removal of each vertex (e.g., as applied in~\cite{qubeV2-2016joong} for computing betweenness centrality of vertices).
\item Use techniques for dynamic shortest paths computation (e.g.,~\cite{qubeV2-2016joong, kas13betweenness,ramalingam92incremental, kourtellis14streamingbetw}) that keep always up-to-date the true distances between vertices, while the subgraph is changing.
\item Use approximation or oracle-based techniques to estimate the shortest distance between two vertices in the subgraph (e.g. as in~\cite{potamias09fast-shortestpaths, qi2013distanceoracle, bergamini2015approximatingbetwn, sommer16approx-shortestpaths}).
If the vertices are close in the graph (e.g., within the same community), such techniques will estimate distances between vertices that are very close to the true distances.
If the vertices are far in the graph, the estimated distances will be high; this output can be a quick hint for the algorithm to avoid trying to connect them.
\item Parallelize the computation (or approximation) of the shortest paths using $|S|$ parallel threads, one for each source in the subgraph under investigation (e.g., as in~\cite{kourtellis14streamingbetw}).
\squishend
These techniques are well studied in the literature and are beyond the scope of the present study.

%% file: tables/datasets.tex
\begin{table}[t!]
\caption{Summary of graphs used. $\delta$: density, ad: average degree, cc: clustering coefficient, ed: effective diameter.\label{tab:all-nets}}
\centering
\vspace{-2mm}
\small
\tabcolsep=0.12cm
\begin{tabular}{lrrrrrrr}
\toprule
\hspace{-4mm}& Dataset      & $|V|$   & $|E|$   & $\delta$  & ad    & cc    & ed  \\
\midrule
\hspace{-4mm}& \dataset{football}$^7$ & 	115		&	613		&	9.4e-2	&	21.3	&	0.40	&	3.9  \\
\hspace{-4mm}& \dataset{kdd14twitter}$^8$       &  1,059   &   2,691 & 4.8e-3 & 5.1 & 0.46 & 8   \\
\hspace{-4mm}& \dataset{flickr}$^{9}$      & 80,513      & 5,899,882     &  1.8e-3& 146.5& 0.17 & 4 \\
\hspace{-4mm}& \dataset{amazon}$^{10}$      & 334,863    & 925,872    &  1.6e-5 &5.5 &0.39 & 15\\
\hspace{-4mm}& \dataset{dblp}$^{10}$       & 317,080    & 1,049,866    &  2.1e-5	&	6.62	&	0.63	&	8.2\\
\hspace{-4mm}& \dataset{youtube}$^{10}$        & 1,138,499    & 2,990,443    & 4.6e-6	 &	5.27	&	0.08	&	6.5\\
\hspace{-4mm}& \dataset{livejournal}$^{10}$       & 3,997,962  & 34,681,189   &  4.3e-6	&	17.3	&	0.28	&	6.5  \\
\bottomrule
\end{tabular}
\vspace{3mm}
\end{table}

%% file: tables/tablesmall.tex

\begin{table*}[t!]
\centering
\caption{Characteristics of the subgraph $H$ extracted:
inefficiency $\mathcal{I}(H)$, number of vertices $|V(H)|$, density $\delta(H)$, average betweenness $bc(V(H)\setminus Q)$ and harmonic centrality $hc(V(H)\setminus Q)$ of the non-query vertices added to the solution. Datasets: small synthetic ($d_1$), \dataset{football}  ($d_2$), and \dataset{kdd14twitter} ($d_3$).
\label{tab:quality1}}
\vspace{-2mm}
\normalsize
\setlength\tabcolsep{2pt}
\begin{tabular}{r|ccc|ccc|ccc|ccc|ccc|}
  \multicolumn{1}{c}{} &
  \multicolumn{3}{c}{$\mathcal{I}(H)$} &
   \multicolumn{3}{c}{$|V(H)|$} &
    \multicolumn{3}{c}{$\delta(H)$} &
    \multicolumn{3}{c}{$bc(V(H)\setminus Q)$} &
    \multicolumn{3}{c}{$hc(V(H)\setminus Q)$} \\
    \cline{2-16}
   & $d_1$ & $d_2$ & $d_3$ & $d_1$ & $d_2$ & $d_3$ & $d_1$ & $d_2$ & $d_3$ &
   $d_1$ & $d_2$ & $d_3$ & $d_1$ & $d_2$ & $d_3$ \\   \cline{2-16}

\mis & \textbf{136.19} & \textbf{107.99} & \textbf{174.08}  
         & \textbf{21.50} &\textbf{20.68} &\textbf{20.68}  
         &\textbf{0.19} &\textbf{0.24} &\textbf{0.23} 
        &\textbf{0.31} &\textbf{0.4} &\textbf{0.36} 
        &\textbf{13.43} & \textbf{10.93} & \textbf{9.06}\\ 

\bh & 158.54 & 120.16 & 215.99  
         & 23.25 & 21.84 & 25.4 
        & 0.17 & 0.22 & 0.22 
        & 0.24 & 0.25 & 0.22 
        & 10.69 & 9.49 & 7.45\\ 
 \mdl & 178.63 &142.93 & 289.07  
         & 21.08 & 20.8 & 26.2
         & 0.08 & 0.09 & 0.09 
        & 0.19 & 0.21 & 0.04 
        & 2.07 & 6.67 & 3.56\\ 
 \cline{2-16}
\end{tabular}
\vspace{3mm}
\end{table*}

%% file: sections/casestudies.tex

In this section, we explore several interesting datasets that act as anecdotal evidence of the utility of {\mis} in the real world, aside from the application with the human connectome described in Section~\ref{subsec:first_comparison}.  All the datasets and contextual information are publicly available.
 For sake of comparison and completeness, we also report the connectors produced by {\bh} and {\mdl} for the same queries.

\spara{Cohesive meal creation.}  Recently there has been increased attention on the science of food, including recipe recommendation~\cite{teng2012recipe} and matching foods with complimentary chemical flavor profiles~\cite{ahn2011flavor}.  In fact, a website was developed for identifying ingredients with profiles similar to a given single ingredient.\footnote{https://www.foodpairing.com/en/home}
The graph is constructed by using the $1$k ingredients as vertices, and adding edges between ingredients that share chemical flavor profiles.\footnote{http://www.nature.com/articles/srep00196} We consider the setting of preparing a meal in a foreign cuisine based on a few candidate ingredients; clearly it is desirable not only to match ingredients that form a delicious meal, but also to avoid using ingredients that do not fit well with all of the others.

 \begin{figure}[H]
   \centering
   \begin{tabular}{ccc}
 \mis & \bh & \mdl \\
 \hspace{-1mm}\includegraphics[scale=.24]{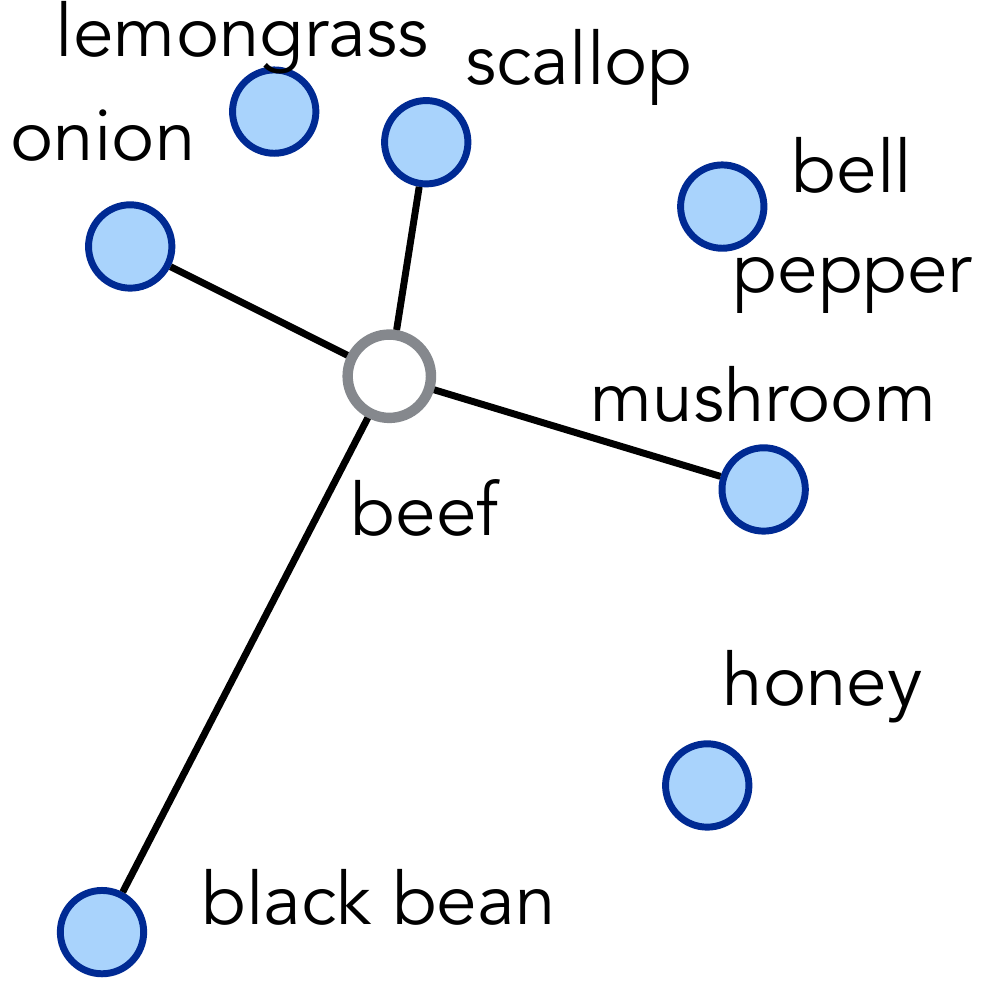} &\
 \hspace{-1mm}\includegraphics[scale=.24]{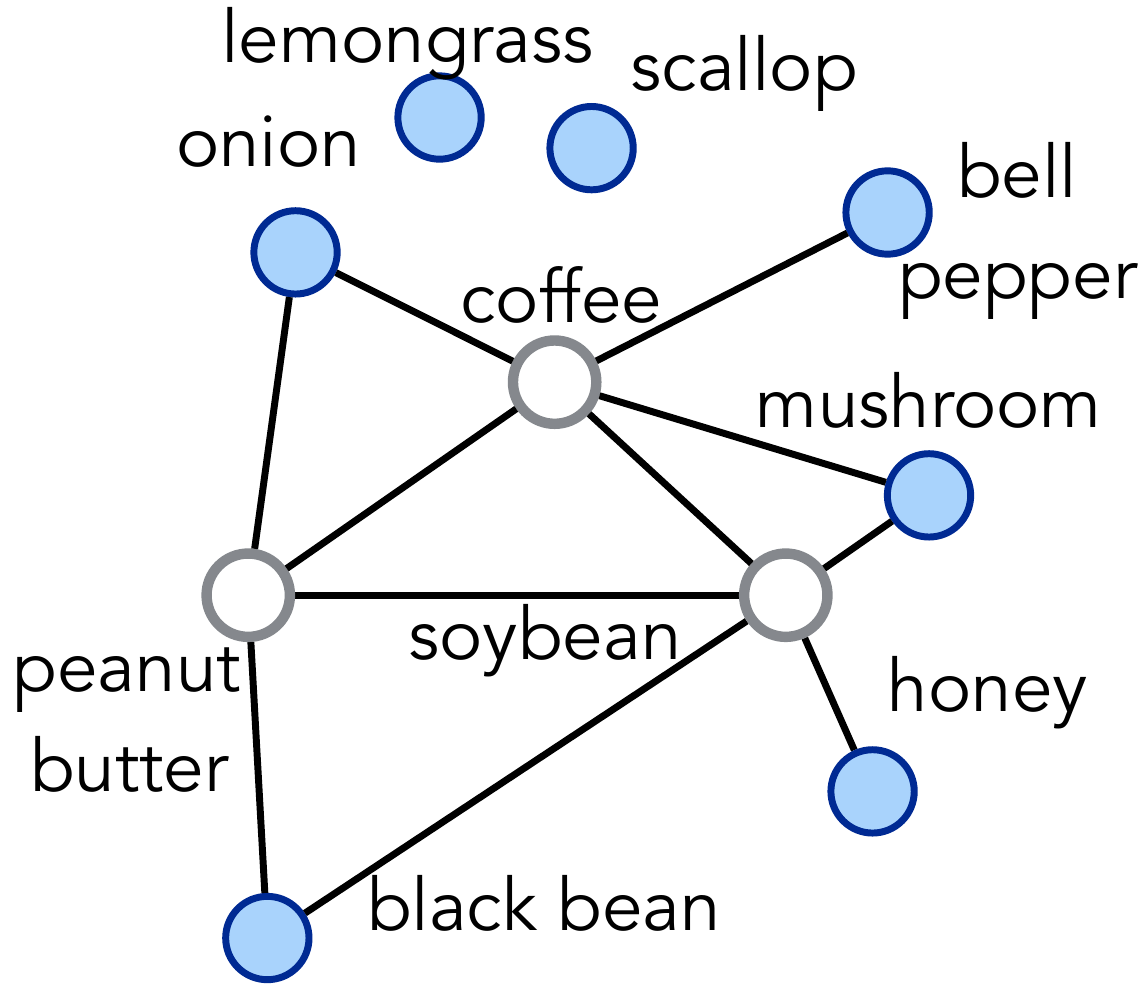} &
\hspace{-1mm} \includegraphics[scale=.24]{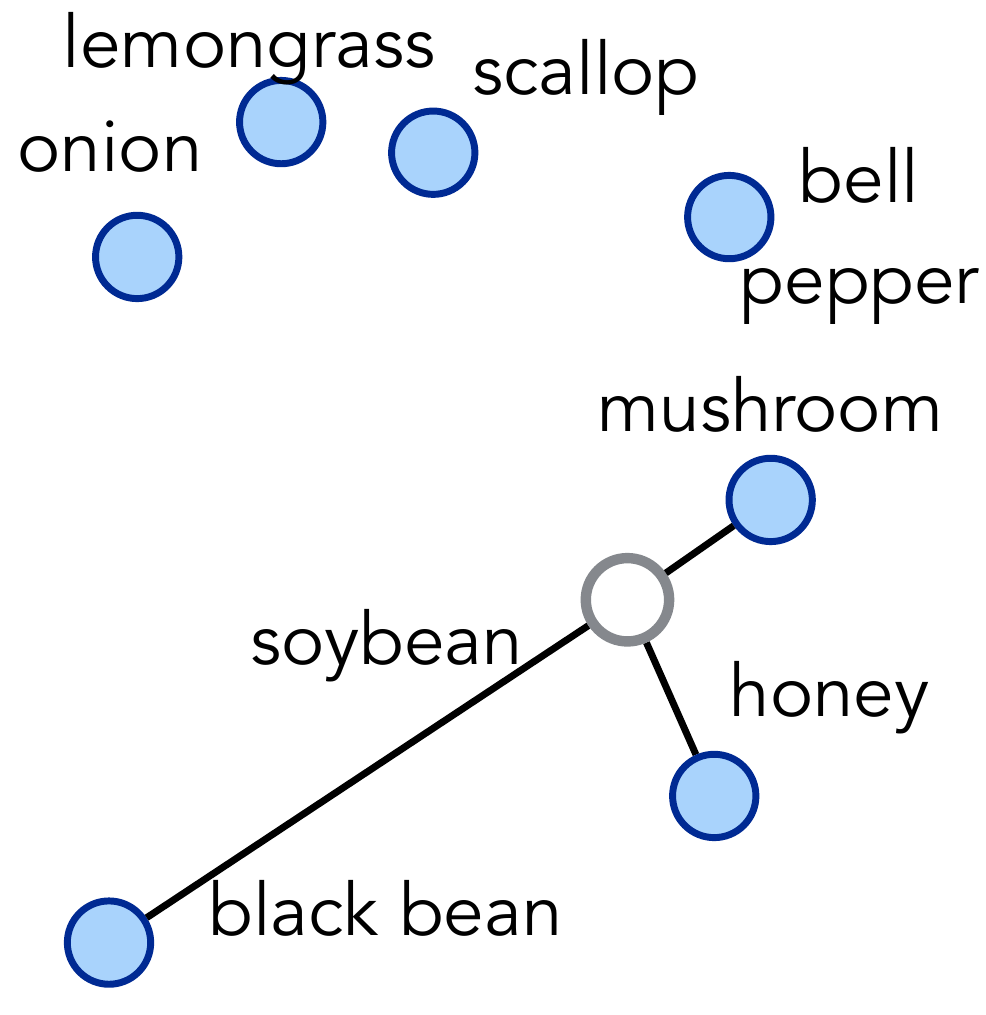}
    \end{tabular}
    \vspace{-4mm}
 \caption{Food network example: query vertices in blue.\label{fig:food}}
\vspace{-2mm}
 \end{figure}

In Figure~\ref{fig:food}, we can see that the {\mis} is obtained by only adding the very central vertex \textsf{beef}, and forming a cohesive meal: onion, beef, scallop, beans, and mushroom. At the same time, ingredients such as honey that do not fit well with the component are disconnected.
In contrast, both {\bh} and {\mdl} suggest unappetizing combinations such as peanut butter and onion, or beans and mushrooms with honey -- missing the simple connector formed by incorporating \textsf{beef} as an ingredient.

\vspace{2mm}

\spara{Functional protein disease association.}  We next consider {\mis} as an aid for biological discovery, as mentioned in Section~\ref{sec:intro}.
We obtain a protein-protein-interaction (PPI) network of $22$k human proteins with edges denoting interactions.\footnote{Data was collected from http://string-db.org/}  We simulate the setting in which a biologist may query a set of vertices and inquire about their possible disease association as a guide for actual lab experimentation.  

 \begin{figure}[H]
   \centering
   \begin{tabular}{ccc}
    \mis & \bh & \mdl \\
 \hspace{-7mm}\includegraphics[scale=.24]{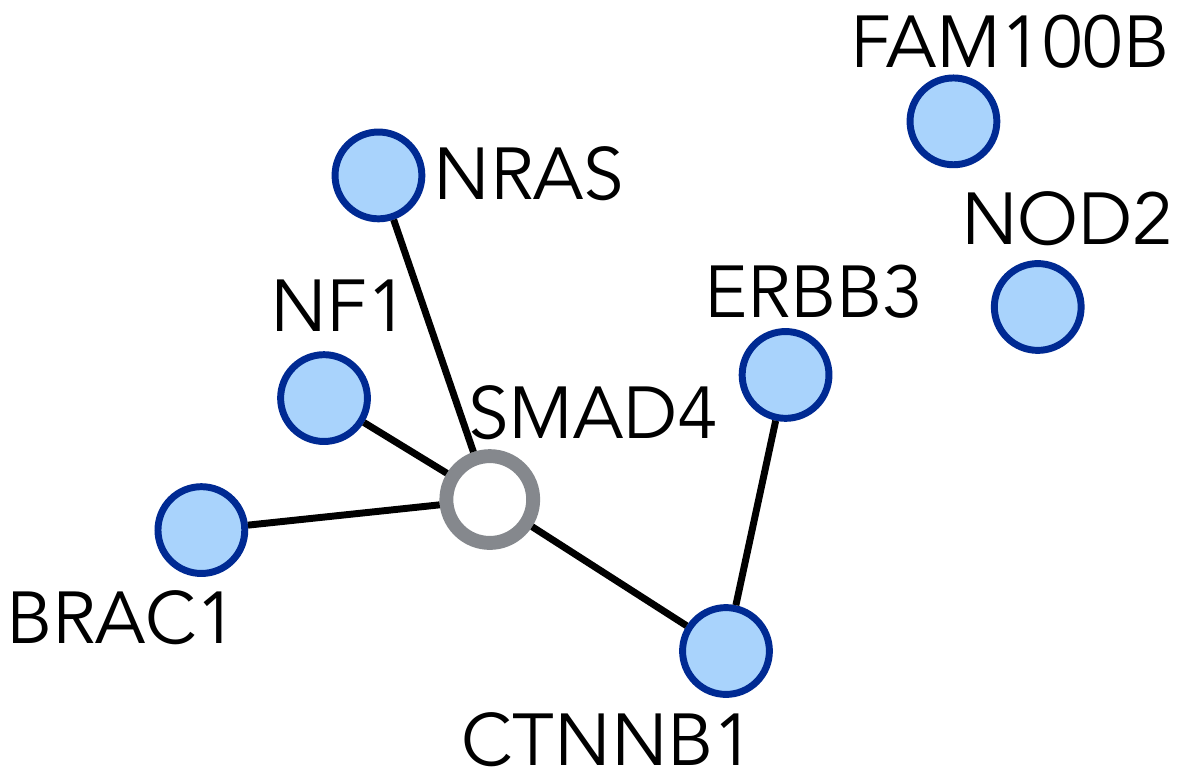} &\
 \hspace{-3mm}\includegraphics[scale=.24]{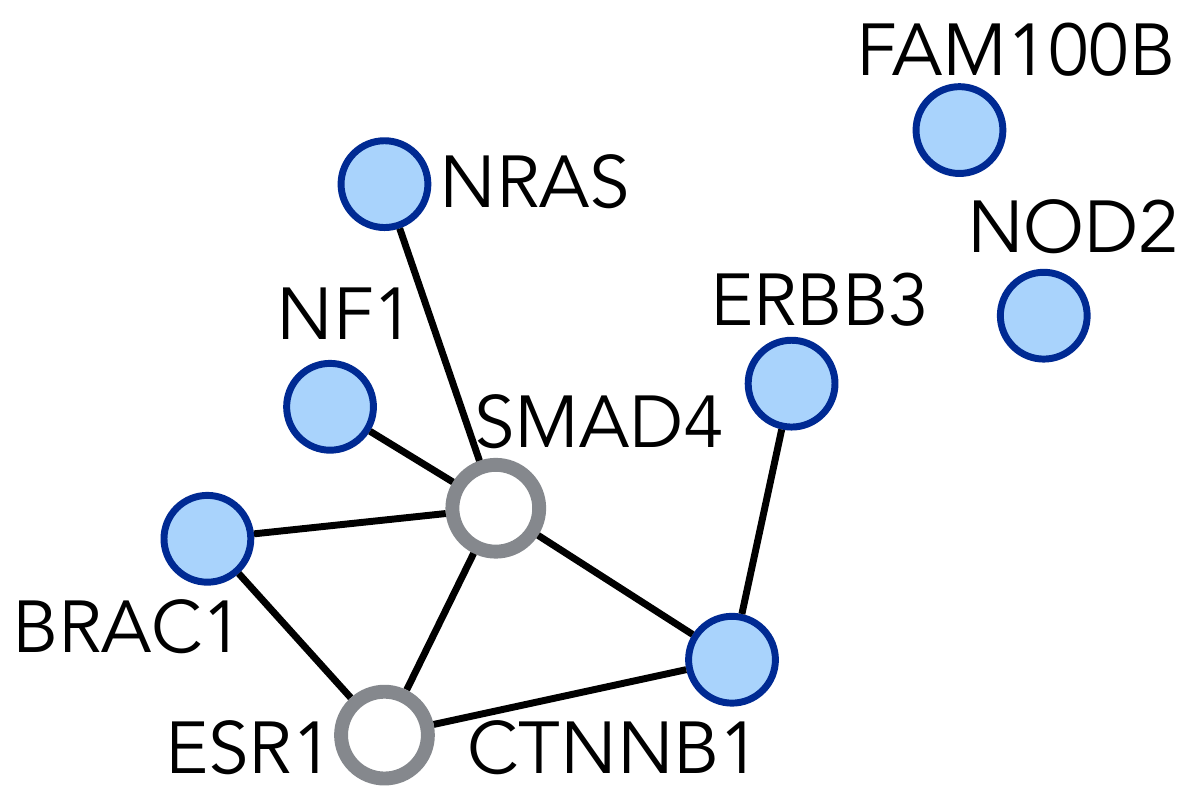} &
\hspace{-3mm} \includegraphics[scale=.24]{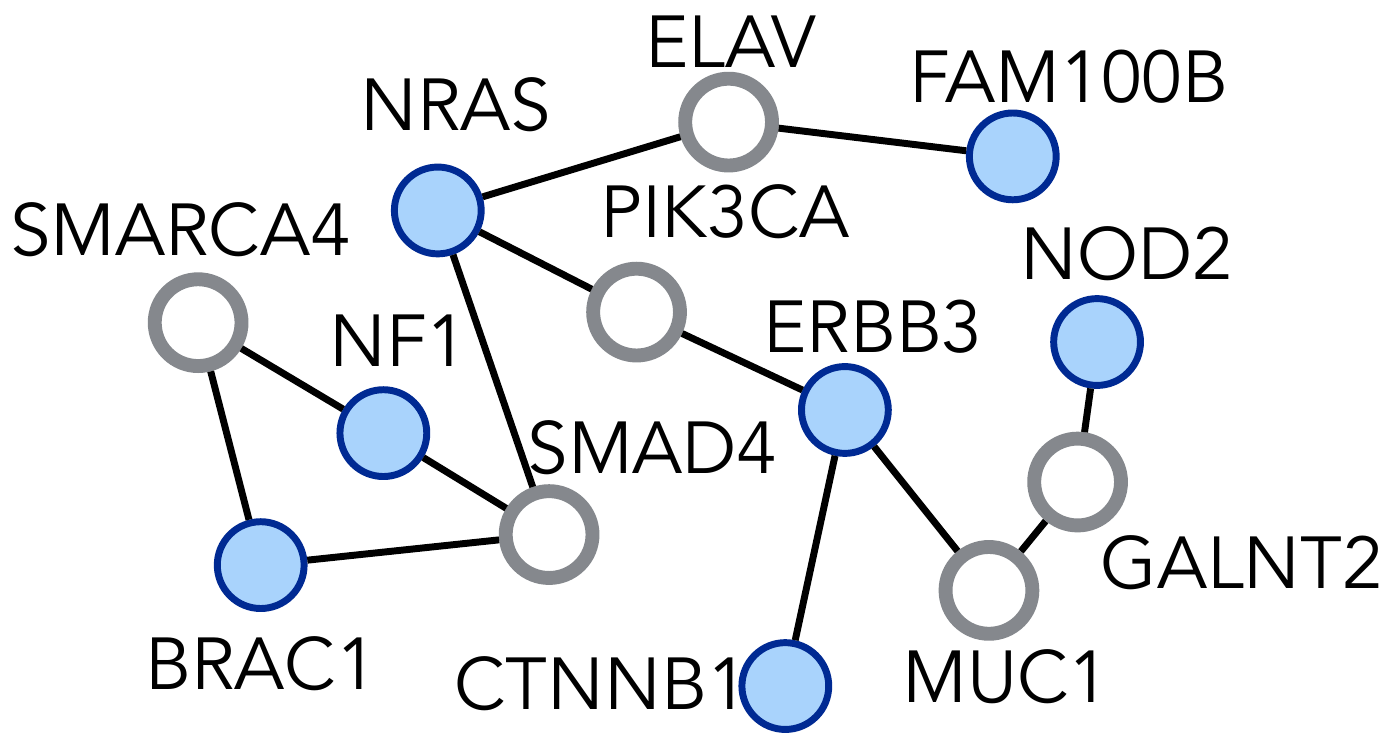}
    \end{tabular}
    \vspace{-4mm}
\caption{PPI network: query vertices in blue.\label{fig:bio}}
\vspace{1mm}
 \end{figure}

In the {\mis} above, the additional vertex added is \textsf{SMAD4} which is known for its role in cancer.  This structure uncovers that, in fact, each query vertex in that connected component has also been linked to cancer.  Meanwhile the disconnected \textsf{NOD2} is associated with Chrons disease, and \textsf{FAM110B} has no strong association.
On the other hand, both {\mdl} and {\bh} add more vertices than necessary, with {\mdl} hiding the association completely.

\spara{Political Stance discovery:}  As a third example, we consider the setting of a social scientist analyzing human interaction with respect to a particular topic.  With a set of users in mind, the scientist seeks insights on how the users interact with each other, and whether strong relationships exist.
The graph is a set of Twitter users that interacted with a particular topic, in our case, the 2016 US election.\footnote{http://www.vertexxlgraphgallery.org/Pages/Graph.aspx?graphID=83188}  

 \begin{figure}[H]
 \vspace{1mm}
   \centering
   \begin{tabular}{ccc}
       \mis & \bh & \mdl \\
 \hspace{-5mm}\includegraphics[scale=.29]{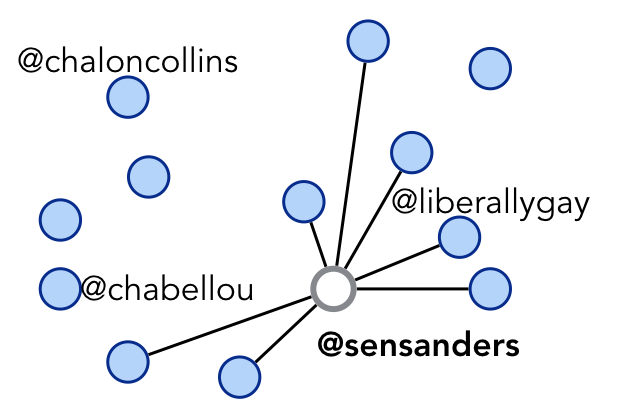} &\
 \hspace{-3mm}\includegraphics[scale=.29]{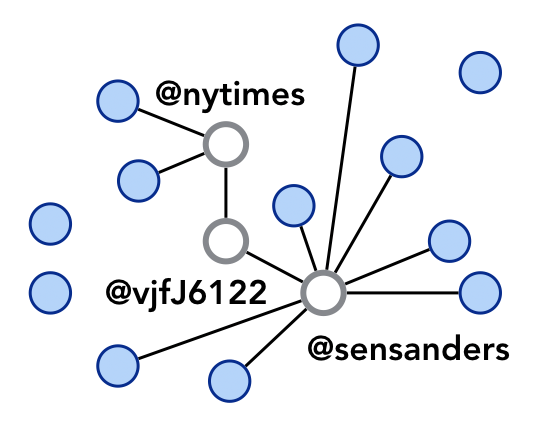} &
\hspace{-3mm} \includegraphics[scale=.29]{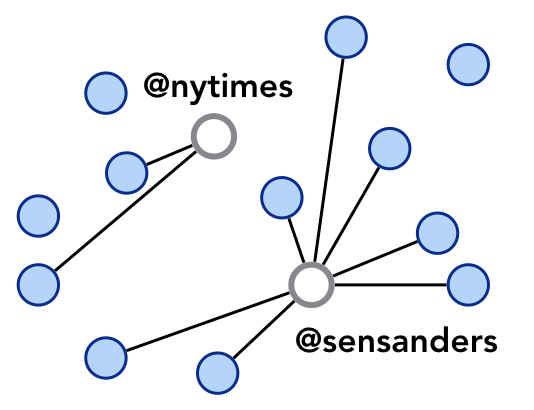}\\
    \end{tabular}
    \vspace{-4mm}
 \caption{Social network: query vertices in blue.\label{fig:drumpf}}
\vspace{1mm}
 \end{figure}

 Above, we see the {\mis} for a random query set.  For each vertex, we check the contextual information available with the graph as well as the Twitter profiles.  We find that the connected component incorporates vertices that express support for Senator Bernie Sanders who is the central added vertex.  The disconnected vertices express other political views, for example, both \texttt{@chaloncollins} and \texttt{@chabellou} express strong support for Donald Trump.  In comparison, both {\mdl} and {\bh} incorporate weaker connections between vertices that hide the natural division of vertices.

\spara{Topic exploration:}  In this example, we consider the educational setting where one is curious about the relationship among a set of topics. For example, which topics are similar?  Or, which topics branched off of others?  Specifically, we consider a graph of $2$k philosophers  with edges are added according to the \emph{influenced-by} section of wikipedia.\footnote{http://www.coppelia.io/2012/06/graphing-the-history-of-philosophy/}

 \begin{figure}[H]
 \vspace{-2mm}
   \centering
   \begin{tabular}{ccc}
       \mis & \bh & \mdl \\
 \hspace{-5mm}\includegraphics[scale=.29]{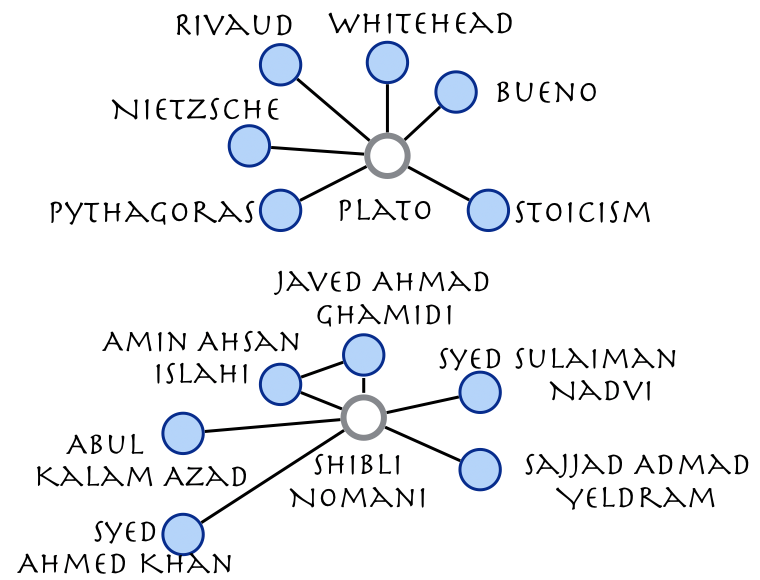} &\
 \hspace{-3mm}\includegraphics[scale=.29]{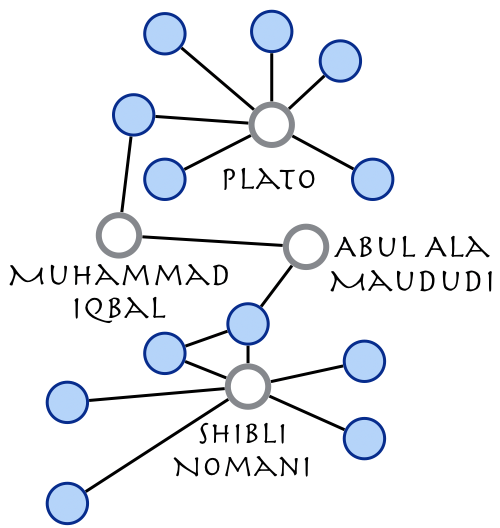} &
\hspace{-3mm} \includegraphics[scale=.29]{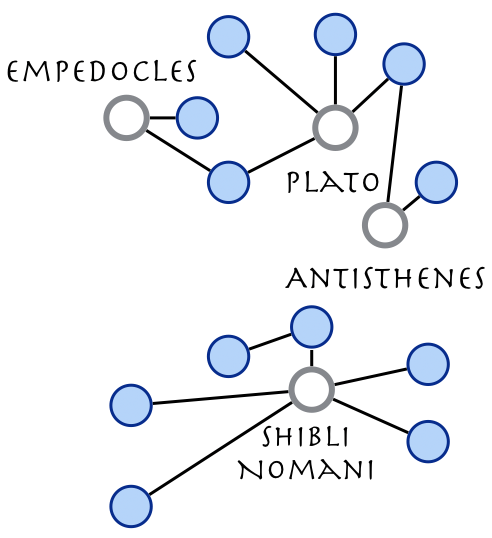}\\
    \end{tabular}
    \vspace{-4mm}
 \caption{Philosophers network: query vertices in blue.\label{fig:phil}}
\vspace{-2mm}
 \end{figure}

We selected a random sample of philosophers, and the resulting {\mis} is shown above.  We see very clearly two connected components, one corresponding to Western philosophers and the other to Islamic philosophers -- with just a glance we understand the ideological relationships among query vertices.  Further, each of the added vertices is a key historical figure in the respective tradition,  summarizing the connections by well-known (central) philosophers.
From the comparison, we see that  {\bh} and {\mdl} are less succinct and less informative about the strong and separate ideologies.

%% file: sections/conclusions.tex
In this paper, we study the general class of problems related to finding a selective connector of a graph:  given a graph and a set of query vertices, find a subgraph that contains all query vertices and optimizes a certain measure of cohesiveness, while also not necessarily requiring the output subgraph to be connected.

In this regard, we define a new graph-theoretic measure, dubbed \emph{network inefficiency}, that allows for simultaneously accounting for both requirements of high-cohesiveness and non-mandatory connectedness.
The specific selective-connector problem instance we tackle in this work is what we call the \emph{minimum inefficiency subgraph} problem, which requires finding a subgraph containing the input query vertices and having minimum network inefficiency.
We show that the problem is \NPhard, and devise a greedy heuristic that provides effective solutions.
We empirically assess the performance of the proposed algorithms in a variety of synthetic and real-world graphs, as well as by several case studies from different domains, such as human brain, cancer, and food networks.

In the future, we plan to extend the portfolio of algorithms for the minimum-inefficiency-subgraph problem by focusing on approximation algorithms with provable quality guarantees and/or heuristics based on paradigms other than greedy.
We also plan to study incremental versions of the problem, where solutions are not recomputed from scratch if a change occurs in the query vertex set and/or the input graph.